\documentclass[envcountsame,envcountsect,orivec,runningheads]{llncs}
\usepackage{all,all_keywords}
\usepackage{hyperref}
\usepackage{amsmath,amsfonts,amssymb,amscd}
\usepackage{algorithm,algorithmic}
\usepackage{latexsym}
\usepackage{enumerate}
\usepackage{graphicx}
\usepackage{xspace}
\usepackage{tikz}
\usepackage{rotating}
\usepackage{xcolor}
\usepackage{colortbl}
\usepackage{wrapfig}

\spnewtheorem*{proofsketch}{Proof sketch}{\itshape}{\rmfamily}

\newcommand{\all}{{\sc ALL}\xspace}
\newcommand{\awn}{{\sc AWN}\xspace}

\newcommand{\Sect}[1]{\autoref{sec:#1}\xspace}
\newcommand{\Fig}[1]{\autoref{fig:#1}\xspace}
\newcommand{\Tab}[1]{\autoref{tab:#1}\xspace}
\newcommand{\Line}[2]{Line~\ref{pro:#1:line#2}}
\DeclareSymbolFont{frenchscript}{OMS}{ztmcm}{m}{n}
\DeclareMathSymbol{\T}{\mathord}{frenchscript}{84}   % translation function
\DeclareMathSymbol{\pow}{\mathord}{frenchscript}{80} % powerset
\newcommand{\NN}{% natural numbers
    \ensuremath{%
        \mathop{\rm I\mkern-2.5mu N}%
        \nolimits%
    }%
}

\newcommand{\incr}{\mbox{\texttt{\,+\!+}}}

\newcommand{\multiAssign}[2][\xi]{
  #1{\mbox{\scriptsize 
  $\left [\begin{array}{@{}l@{}}
		#2
  \end{array}\right ]$
}}}

\newcommand{\undefined}{\mathord\uparrow}

\newcommand{\plat}[1]{\raisebox{0pt}[0pt][0pt]{#1}} % no vertical space

\newcommand{\sidecond}[1]{\mbox{\scriptsize  
$\left(\begin{array}{@{}l@{}}
		#1
 \end{array}\right)$}}

\makeatletter
\def\moverlay{\mathpalette\mov@rlay}
\def\mov@rlay#1#2{\leavevmode\vtop{%
   \baselineskip\z@skip \lineskiplimit-\maxdimen
   \ialign{\hfil$\m@th#1##$\hfil\cr#2\crcr}}}
\newcommand{\charfusion}[3][\mathord]{
    #1{\ifx#1\mathop\vphantom{#2}\fi
        \mathpalette\mov@rlay{#2\cr#3}
      }
    \ifx#1\mathop\expandafter\displaylimits\fi}
\makeatother
\newcommand{\dcup}{\charfusion[\mathbin]{\cup}{\mbox{\Large$\cdot$}}}

\newcommand{\when}{{\rm if}~}
\newcommand{\cinc}{\texttt{c+}}
\newcommand{\bis}{\ensuremath{\mathop{\,\raisebox{.3ex}{$\underline{\makebox[.7em]{$\leftrightarrow$}}$}\xspace}}}

\newcommand{\application}{network\xspace}
\newcommand{\durAck}{\constant{dur\_ack}}%
\newcommand{\p}{P}
\newcommand{\q}{Q}
\newcommand{\rts}{RTS\xspace}
\newcommand{\cts}{CTS\xspace}

\newcommand{\uppaal}{\textsc{Uppaal}\xspace}
\newcommand{\prism}{PRISM\xspace}
\newcommand{\apmc}{APMC\xspace}

\newcommand{\dpkt}{frame}
\newcommand{\pkt}{packet}
\newcommand{\chunks}{chunks}
\newcommand{\chnk}{chunk}

\makeatletter
\def\comesfrom{\@transition\leftarrowfill}
\def\goesto{\@transition\rightarrowfill}
\def\ngoesto{\@transition\nrightarrowfill}
\def\Goesto{\@transition\Rightarrowfill}
\def\nGoesto{\@transition\nRightarrowfill}
\def\xmapsto{\@transition\mapstofill}
\def\nxmapsto{\@transition\nmapstofill}
\def\@transition#1{\@@transition{#1}}
\newbox\@transbox
\newbox\@arrowbox
\newbox\@downbox
\def\@@transition#1#2%
   {\setbox\@transbox\hbox
      {\vrule height 1.5ex depth .9ex width 0ex\hskip0.25em$\scriptstyle#2$\hskip0.25em}
   \ifdim\wd\@transbox<1.5em
      \setbox\@transbox\hbox to 1.5em{\hfil\box\@transbox\hfil}\fi
   \setbox\@arrowbox\hbox to \wd\@transbox{#1}
   \ht\@arrowbox\z@\dp\@arrowbox\z@
   \setbox\@transbox\hbox{$\mathop{\box\@arrowbox}\limits^{\box\@transbox}$}
   \dp\@transbox\z@\ht\@transbox 10pt
   \mathrel{\box\@transbox}}
\def\nrightarrowfill{$\m@th\mathord-\mkern-6mu%
  \cleaders\hbox{$\mkern-2mu\mathord-\mkern-2mu$}\hfill
  \mkern-6mu\mathord\not\mkern-2mu\mathord\rightarrow$}
\def\Rightarrowfill{$\m@th\mathord=\mkern-6mu%
  \cleaders\hbox{$\mkern-2mu\mathord=\mkern-2mu$}\hfill
  \mkern-6mu\mathord\Rightarrow$}
\def\nRightarrowfill{$\m@th\mathord=\mkern-6mu%
  \cleaders\hbox{$\mkern-2mu\mathord=\mkern-2mu$}\hfill
  \mkern-6mu\mathord\not\mathord\Rightarrow$}
\def\mapstofill{$\m@th\mathord\mapstochar\mathord-\mkern-6mu%
  \cleaders\hbox{$\mkern-2mu\mathord-\mkern-2mu$}\hfill
  \mkern-6mu\mathord\rightarrow$}
\def\nmapstofill{$\m@th\mathord\mapstochar\mathord-\mkern-6mu%
  \cleaders\hbox{$\mkern-2mu\mathord-\mkern-2mu$}\hfill
  \mkern-6mu\mathord\not\mkern-2mu\mathord\rightarrow$}
\makeatother %%%%% end of arrows definition
\newcommand{\ar}[1]{\mathrel{\goesto{#1}}}            % arrow
\newcommand{\nar}[1]{\mathrel{\ngoesto{#1\;}}}        % negated arrow

\title{A Process Algebra for Link Layer Protocols}
\authorrunning{R.J. van Glabbeek, P. H\"ofner and M. Markl}
  \author{
  Rob van Glabbeek \inst{1,2},
  Peter H\"ofner\inst{1,2}\and
  Michael Markl\inst{1,3}
  }
\institute{
  Data61, CSIRO, Australia\and
  Computer Science and Engineering, University of New South Wales, Australia\and
  Institut f\"ur Informatik, Universit\"at Augsburg, Germany
}

\begin{document}%
\maketitle 
\setcounter{footnote}{0}
  \renewcommand{\theoremautorefname}{Theorem}
  \renewcommand{\tableautorefname}{Table}
  \renewcommand{\figureautorefname}{Figure}
  \renewcommand{\sectionautorefname}{Section}
  \renewcommand{\subsectionautorefname}{Section}
  \renewcommand{\subsubsectionautorefname}{Section}
  \newcommand{\algorithmautorefname}{Process}

\begin{abstract}
We propose a process algebra for link layer protocols, featuring a unique mechanism for modelling frame collisions.
We also formalise suitable liveness properties for link layer protocols specified in this framework. 
To show applicability we model and analyse two versions of  the Carrier-Sense Multiple Access with Collision Avoidance (CSMA/CA) protocol. 
Our analysis confirms the hidden station problem for the version without virtual carrier sensing.
However, we show that the version with virtual carrier sensing not only overcomes
this problem, but also the exposed station problem with probability $1$.
Yet the protocol cannot guarantee {\pkt} delivery, not even with probability~$1$.
\end{abstract}

%%%%%%%%%%%%%%%%%%%%%%%%%%%%%%%%%%
\section{Introduction}\label{sec:introduction}
%%%%%%%%%%%%%%%%%%%%%%%%%%%%%%%%%%
The (data) link layer is the 2nd layer of the ISO/OSI model of computer networking~\cite{ISOIEC7498}.
Amongst others, it is responsible for the transfer of data between adjacent nodes in Wide Area Networks (WANs)
and Local Area Networks~(LANs).
\newcommand{\oxford}[1]{}

Examples of link layer protocols are Ethernet for
LANs~\cite{IEEE8023}, the Point-to-Point Protocol~\cite{rfc1661}\oxford, and the High-Level Data Link Control protocol (e.g.\ \cite{Friend88}).
Part of this layer are also multiple access protocols  such as the Carrier-Sense Multiple Access with Collision Detection (CSMA/CD) protocol 
for re-transmission in Ethernet bus networks and hub networks, or the Carrier-Sense Multiple Access with Collision Avoidance 
(CSMA/CA) protocol~\cite{IEEE80211,IEEE80215} in wireless networks.

One of the unique characteristics of the link layer is that when devices attempt to use a medium simultaneously, \emph{collisions of messages}
occur. 
So, any  modelling language and formal analysis of layer-2 protocols has to support such collisions.
Moreover, some protocols are of probabilistic nature: CSMA/CA for example chooses time slots probabilistically with discrete uniform distribution.

As we are not aware of any formal framework with primitives for modelling data collisions,
this paper introduces a process algebra
for modelling and analysing link layer protocols.
In \Sect{pa_1} we present an algebra featuring a unique mechanism for modelling collisions, `hard-wired' in the semantics.
It is
the nonprobabilistic fragment of the Algebra for Link Layer protocols (\all), which we introduce in \Sect{pa_2}. 
In \Sect{liveness} we formulate \emph{packet delivery}, a liveness property that ideally ought to hold
for link layer protocols, either outright, or with a high probability.
In \Sect{csma} we use this framework to formally
model and analyse the CSMA/CA protocol.\pagebreak[3]

Our analysis confirms the hidden station problem for the version of CSMA/ CA 
without virtual carrier sensing (\autoref{sec:Hidden}).
However, we also show that the version with virtual carrier sensing overcomes not only this problem, but also the exposed station problem with probability $1$.
Yet the protocol cannot guarantee {\pkt} delivery, not even with probability $1$.

%%%%%%%%%%%%%%%%%%%%%%%%%%%%%%%%%%
\section{A Non-Probabilistic Subalgebra}\label{sec:pa_1}
%%%%%%%%%%%%%%%%%%%%%%%%%%%%%%%%%%
In this section we propose a timed process algebra that can model the collision of 
link layer messages, called \emph{frames}.%
\footnote{As it is the nonprobabilistic fragment of a forthcoming algebra we do not name it.}
It can be used for link layer protocols that do not feature probabilistic choice, and
is inspired by the (Timed) Algebra for Wireless Networks (({\sc T}-)\awn)~\cite{ESOP12,TR13,ESOP16}, a
process algebra suitable for modelling and analysing protocols on layers 3 (network) and 4 (transport)
of the OSI model.

The process algebra models a (wired or wireless) network as an encapsulated parallel composition of network nodes. 
Due to the nature of the protocols under consideration, on each node exactly one sequential process is running.
The algebra features a discrete model of time, where each sequential process maintains a local variable
\now\ holding its local clock value---an integer. We employ only one clock for each sequential
process. All sequential processes in a network synchronise in taking time steps, and at each time
step all local clocks advance by one unit. Since this means that all clocks are in sync and do not run at different speeds
it is clear that we do not consider the problem of clock shift.
For the rest, the variable \now\ behaves like
any other variable maintained by a process: its value can be read when evaluating guards, thereby
making progress time-dependant, and any value can be assigned to it, thereby resetting the local
clock. Network nodes communicate with their direct neighbours---those nodes that are in transmission
range.  The algebra provides a mobility option that allows nodes to move in or out of transmission
range. The encapsulation of the entire network inhibits communications 
between network nodes and the outside world, with the exception of the 
receipt and delivery of data packets from or to clients (the higher OSI layers).

\subsection{A Language for Sequential Processes}
The internal state of a process is determined, in part, by the values of certain data variables that are maintained by that process.  
To this end, we assume a data structure with several types, variables ranging over these types, operators and predicates. 
Predicate logic yields terms (or \emph{data expressions}) and formulas to denote data values and statements about them. 
Our data structure always contains the types
\tTIME,
\tDATA,
\tMSG,
\tCHUNK,
$\tIP$ and
$\pow{(\tIP)}$
of discrete \emph{time values}, which we take to be integers, \emph{\application layer data},
\emph{messages}, \emph{chunks} of messages that take one time unit to transmit,
\emph{node identifiers} and \emph{sets of node identifiers}.
We further assume that there are variables \now\ of type {\tTIME} and $\rfr$ of type \tCHUNK.
In  addition, we assume a set of \emph{process names}.\pagebreak[4]
Each process name $X$ comes with a \emph{defining equation}\vspace{-1.6ex}
\[
X(\keyw{var}_1,\ldots,\keyw{var}_n) \stackrel{{\it def}}{=}  \p\ ,\vspace{-2pt}
\]
in which $n\in\NN$, $\keyw{var}_i$ are variables and $\p$ is a \emph{sequential process expression} defined by the grammar below. 
It may contain the variables $\keyw{var}_i$ as well as $X$. 
However, all occurrences of data variables in $\p$ have to be \emph{bound}.%
\footnote{An occurrence of a data variable in $\p$ is \emph{bound} if it is one of the variables
  $\keyw{var}_i$, one of the two special variables $\now$ or $\rfr$, a variable \keyw{var} occurring
  in a subexpression $\assignment{\keyw{var}\mathop{:=}\dexp{exp}}\q$, an occurrence in a
  subexpression $\cond{\varphi}\q$ of a variable occurring free in $\varphi$, or
  a variable \keyw{data} or \keyw{dest} occurring in a subexpression
  $\textbf{newpkt}(\keyw{data},\keyw{dest}).\q$. Here $\q$ is an arbitrary sequential process expression.}
The choice of the underlying data structure and the process names with their defining equations can be tailored to any particular application of our language.

The \emph{sequential process expressions} are given by the following grammar:
\begin{eqnarray*}
\SP &::=& 
	X(\dexp{exp}_1,\ldots,\dexp{exp}_n) ~\mid~ \cond{\varphi}\SP ~\mid~ \assignment{\keyw{var}:=\dexp{exp}} \SP ~\mid~ \alpha.\SP ~\mid~ \SP+\SP\\
	\alpha &::=& \transmit{\dexp{ms}} ~\mid~ \textbf{newpkt}(\keyw{data},\keyw{dest})~\mid~ \deliver{\keyw{data}}
\end{eqnarray*}
Here $X$ is a process name, $\dexp{exp}_i$ a data expression of the same type as
  $\keyw{var}_i$, $\varphi$ a data formula, $\keyw{var}\mathop{:=}\dexp{exp}$ an assignment of a
  data expression \dexp{exp} to a variable \keyw{var} of the same type,
  \dexp{ms} a data expression  of type {\tMSG}, and
  \keyw{data}, \keyw{dest} data variables of types {\tDATA}, {\tIP} respectively.

Given a valuation of the data variables by concrete data values, the
sequential process $\cond{\varphi}\p$ acts as $\p$ if $\varphi$
evaluates to {\tt true}, and deadlocks if $\varphi$ evaluates to
{\tt false}. In case $\varphi$ contains free variables that are not
yet interpreted as data values, values are assigned to these variables
in any way that satisfies $\varphi$, if possible.
The process $\assignment{\keyw{var}\mathop{:=}\dexp{exp}}\p$
acts as $\p$, but under an updated valuation of the data variable~\keyw{var}.
The process $\p\mathop{+}\q$ may act either as $\p$ or as
$\q$, depending on which of the two processes is able to act at all.  In a
context where both are able to act, it is not specified how the choice
is made. The process $\alpha.\p$ first performs the action
$\alpha$ and subsequently acts as $\p$.
The above behaviour is identical to \awn, and many other standard process algebras.
The action $\transmit{\dexp{ms}}$ transmits (the data value bound to the expression) $\dexp{ms}$ to all other network nodes within transmission range.
The action $\textbf{newpkt}(\keyw{data},\keyw{dest})$ models the injection by the \application
  layer of a data packet $\keyw{data}$ to be transmitted to a destination $\keyw{dest}$.
Technically, $\keyw{data}$ and $\keyw{dest}$ are variables that will be bound to the obtained values
upon receipt of a $\textbf{newpkt}$. 
Data is delivered to the \application layer by \deliver{\keyw{data}}.
In contrast to \awn, we do not have a primitive for receiving messages from neighbouring nodes, because our
  processes are \emph{always} listening to neighbouring nodes, in parallel with anything else they do.

As in \awn, the internal state of a sequential process described by an expression
$\p$ is determined by~$\p$, together with a \emph{valuation} $\xi$
associating values $\xi(\keyw{var})$ to variables \keyw{var}
maintained by this process. Valuations naturally extend to
\emph{$\xi$-closed} expressions---those in which all variables are
either bound or in the domain of~$\xi$.
We denote the valuation that assigns the value $v$ to the variable \keyw{var}, and agrees with $\xi$
on all other variables, by $\xi[\keyw{var}:= v]$.\pagebreak[3]
The valuation $\xi_{|S}$
agrees with $\xi$ on all variables $\keyw{var}\in S$ and is undefined otherwise.\pagebreak[4]
Moreover we

\mbox{}\vspace{-22pt} % to correct LLNCS spacing error

\noindent
use $\xi[\keyw{var}\incr]$ as an abbreviation for
$\xi[\keyw{var}:=\xi(\keyw{var})\mathop{+}1]$, for
suitable types.

To capture the durational nature of transmitting a message between network nodes,
we model a message as a sequence of \emph{\chunks}, each of which 
takes one time unit to transmit.
The function $\fndur\mathop: \tMSG \rightarrow \tTIME_{>0}$ calculates the amount of time steps needed for a sending a message, i.e.\ it calculates the number of \chunks.\linebreak[4] 
We employ the internal data type
$\tCHUNK :=  \{\chunk mc\mid m\in\tMSG,1\leq c\leq\dur{m}\} \cup \{\conflict,\medIdle\}$.
The {\chnk} $\chunk mc$ indicates the $c$\,th fragment of a message $m$.
Data conflicts---junk transmitted via the medium---is modelled by the special chunk $\conflict$,
and the absence of an incoming chunk is modelled by $\medIdle$.

Our process algebra maintains a variable $\rfr$ of type $\tCHUNK$,
storing the fragment of the current message received so far.
{\makeatletter
\let\par\@@par
\par\parshape0
\everypar{}\begin{wrapfigure}[10]{r}{0.352\textwidth}
 \vspace{-8ex}
\[\begin{array}{@{\hspace{-2pt}}c|c|c@{}}
\rfr & \dval{ch} & \rfr\star\dval{ch} \\
\cline{1-3}
  *          & \conflict & \conflict \\
  *          & \medIdle  & \medIdle  \\
  *          & \chunk m1 & \chunk m1 \\
\chunk mc    & \chunk m{c{+}1} & \chunk m{c{+}1} \\
\dexp{rfr}   & \chunk m{c{+}1} & \conflict \\ && \hfill \mbox{ if~} \dexp{rfr}\neq\chunk mc\!\!\!\!\\
\end{array}\]
\end{wrapfigure}
\noindent
As a value of this variable, $\chunk mc$ indicates that the first $c$ {\chnk}s
of message $m$ have been received in order; $\conflict$ indicates that the last incoming {\chnk} was not the expected (next) 
part of a message in progress, and  $\medIdle$ indicates that the channel was idle during the last
time step. The table on the right, with $*$ a wild card, shows how the value of $\rfr$ evolves upon receiving a
new {\chnk}$~\dval{ch}$.
\par}

Specifications may refer to the data type $\tCHUNK$ only through the Boolean functions
$\new$---having a single argument \dval{msg} of type $\tMSG$---and $\idle$, defined by
$\new(\dval{msg}) := (\rfr=(\dval{msg}:\dur{\dval{msg}})$
and $\idle := (\keyw{rfr}=\medIdle)$.
A guard $[\new(\dval{msg})]$ evaluates to true iff a new message $\dval{msg}$ has just been received;
$[\idle]$ evaluates to true iff in the last time slice the medium was idle.

The structural operational semantics of \Tab{sos sequential} describes how one internal state can evolve into another by performing an \emph{action}.
The set $\act$ of actions consists of
$\transmitAct{m}{c}{\dval{ch}}$,
$\w(\dval{ch})$,
$\textbf{newpkt}(\dval{d},\dval{dest})$,
$\deliver{\dval{d}}$,
and internal actions~$\tau\!$, for each choice of $m \mathop{\in}\tMSG$, 
$c\mathop{\in}\{1,\dots,\dur{m}\}$,
$\dval{ch}\in\tCHUNK$,
$\dval{d}\mathop{\in}\tDATA$ and
$\dval{dest}\mathop{\in}\tIP$,
where the first two actions are time consuming.
On every time-consuming action, each process receives a {\chnk} $\dval{ch}$ and updates the variable
$\rfr$ accordingly; moreover, the variable $\now$ is incremented on all
process expressions in a (complete) network synchronously.

Besides the special variables $\now$ and $\rfr$, the formal semantics employs an internal variable 
$\counter\mathop\in\NN$ that enumerates the {\chnk}s of split messages and is used to identify which
{\chnk} needs to be sent next.
The variables \now, \rfr\ and \counter\ are not meant to be changed by \all specifications, e.g.\ by using  assignments.
We call them read-only and collect them in the set $\READONLY=\{\now, \rfr, \counter\}$.

Let us have a closer look at the rules of \Tab{sos sequential}.
\begin{sidewaystable}[ph!]
\vspace{-5pt}
{\small
\[\begin{array}{@{}lr@{~\hspace{-.5pt}}c@{~\hspace{-.5pt}}l@{~\hspace{-.5pt}}r}
(1)~~~~ & 
\xi,\transmit{\dexp{ms}}.\p 
	& \ar{\transmitAct{\xi(\dexp{ms})}{\cinc}{\dval{ch}}} 
	& \multiAssign[\xi]{
                        \counter\incr\\[-1pt]
			\rfr:=\rfr\star\dval{ch}\\[-1pt]
			\now\incr
		},\transmit{\xi(\dexp{ms})}.\p
    &\begin{array}{@{}r@{}}
    		\sidecond{\when\cinc < \dur{\xi(\dexp{ms})}} \\
	        \sidecond{\forall \dval{ch}\in\tCHUNK }
     \end{array}
\\[6pt]
%	% Finish transmit 
\rowcolor[gray]{.9}
(2) &
\xi,\transmit{\dexp{ms}}.\p
	& \ar{\transmitAct{\xi(\dexp{ms})}{\cinc}{\dval{ch}}}
	& \multiAssign{
			\counter:=0\\[-1pt]
			\rfr:=\rfr\star\dval{ch}\\[-1pt]
			\now\incr
		},\p
    &\begin{array}{@{}r@{}}
    		\sidecond{\when\cinc = \dur{\xi(\dexp{ms})}} \\
	        \sidecond{\forall \dval{ch}\in\tCHUNK }
     \end{array}
\\[6pt]
(3) &
%	% Receiving app message %
\xi,\textbf{newpkt}(\keyw{data},\keyw{dest}).\p
	& \ar{\obtain{d}{\dval{dest}}}
	& \multiAssign{\keyw{data}:=d\\[-1pt]\keyw{dest}:=\dval{dest}},\p
    &\begin{array}{@{}r@{}} 
		\sidecond{ \forall d \in\tDATA,~ \dval{dest}\in\tIP}
	\end{array}
\\[6pt]
\rowcolor[gray]{.9}
%	% Not receiving app message %
(4) &
\xi,\textbf{newpkt}(\keyw{data},\keyw{dest}).\p
	& \ar{\w(\dval{ch})}
        &        \multiAssign{
				\rfr:=\rfr\star\dval{ch}\\[-1pt]
				\now\incr
			},\textbf{newpkt}(\keyw{data},\keyw{dest}).\p
	&        \sidecond{\forall \dval{ch}\in\tCHUNK }
\\[6pt]
(5) &
%	% Deliver data %
\xi,\deliver{\keyw{data}}.\p
	& \ar{\deliver{\xi(\keyw{data})}}
	& \xi,\p
\\[6pt]
\rowcolor[gray]{.9}
%	% Variable assignment %
(6) &
\xi,\assignment{\keyw{var}:=\dexp{exp}}\p
	& \ar{\tau}
	& \multiAssign{\keyw{var}:=\xi(\dexp{exp})},\p
        &
\\[6pt]
(7) &
%	% Process call %
\multicolumn{3}{c}{\displaystyle
		\frac
		{\multiAssign[\xi_{|\READONLY}]{\keyw{var}_i:=\xi(\dexp{exp}_i)}_{i=1}^n,\p
			\ar{a}
			\zeta,\p'}
		{\xi,X(\dexp{exp}_1,\ldots,\dexp{exp}_n)
			\ar{a}
			\zeta,\p'}
		~\mbox{(\scriptsize$X(\keyw{var}_1,\ldots,\keyw{var}_n) \stackrel{{\it def}}{=} \p$)}
	}
              &\sidecond{\forall a\in\act -\{\w(\dval{ch})\mid \dval{ch}\in\tCHUNK\}}
\\[12pt]
%	% Process call %
\rowcolor[gray]{.9}
(8) &
\multicolumn{3}{c}{\displaystyle
		\frac
		{\multiAssign[\xi_{|\READONLY}]{\keyw{var}_i:=\xi(\dexp{exp}_i)}_{i=1}^n,\p
			\ar{\w(\dval{ch})}
			\zeta,\p'}
		{\xi,X(\dexp{exp}_1,\ldots,\dexp{exp}_n)
			\ar{\w(\dval{ch})}
                 \multiAssign{
				\rfr:=\rfr\star\dval{ch}\\[-1pt]
				\now\incr
			},X(\dexp{exp}_1,\ldots,\dexp{exp}_n)}
		~\mbox{(\scriptsize$X(\keyw{var}_1,\ldots,\keyw{var}_n) \stackrel{{\it def}}{=} \p$)}
	}
	&\sidecond{\forall \dval{ch}\in\tCHUNK }
\\[6pt]
(9) &
\xi,\p
	& \ar{\w(\dval{ch})}
	& \multiAssign{
		\rfr:=\rfr\star\dval{ch}\\[-1pt]
		\now\incr
		},\p
& \begin{array}{@{}r@{}} \sidecond{\when\xi(\p)\undefined} \\
	\sidecond{\forall \dval{ch}\in\tCHUNK }
  \end{array}
\\[8pt]
\rowcolor[gray]{.9}
(10) &
\multicolumn{3}{c}{\displaystyle
		\frac{\xi,\p \ar{a} \zeta,\p'}{\xi,\p+\q \ar{a} \zeta,\p'} \quad
		\frac{\xi,\q \ar{a} \zeta,\q'}{\xi,\p+\q \ar{a} \zeta,\q'}
              }
              &\sidecond{\forall a\in\act -\{\w(\dval{ch})\mid \dval{ch}\in\tCHUNK\}}
\\[12pt]
(11) &
\multicolumn{3}{c}{\displaystyle
		\frac
		{\xi,\p \ar{\w(\dval{ch})} \zeta,\p' \quad \xi,\q \ar{\w(\dval{ch})} \zeta',\q'}
		{\xi,\p+\q \ar{\w(\dval{ch})} \zeta,\p'+\q'}
               }
               &\sidecond{\forall \dval{ch}\in\tCHUNK}
\\[8pt]
\rowcolor[gray]{.9}
(12) &
%%	% Guards %
\multicolumn{3}{c}{\displaystyle
		\frac
		{\xi \stackrel{\varphi}{\rightarrow}\zeta}
		{\xi,\cond{\varphi}\p
			\ar{\tau}
			\zeta,\p} 
			\qquad
		\frac
		{\xi \nar{\varphi}}
		{\xi,\cond{\varphi}\p 
			\ar{\w(\dval{ch})}
			\multiAssign{
				\rfr:=\rfr\star\dval{ch}\\[-1pt]
				\now\incr
			},\cond{\varphi}\p}}
	& \sidecond{\forall \dval{ch}\in\tCHUNK}
\end{array}\]
} % end small
\vspace{-7pt}
\caption{\em Structural operational semantics for sequential process expressions}
\label{tab:sos sequential}
\end{sidewaystable}

The first two rules describe the sending of a message $\dexp{ms}$.
Remember that $\dur{\dexp{ms}}$ calculates the time needed to send $\dexp{ms}$.
The counter $\counter$ keeps track of the time passed already.
The action $\transmitAct{m}{c}{\dval{ch}}$ occurs when the node transmits the fragment $\chunk{m}{c}$;
simultaneously, it receives the fragment $\dval{ch}$.\footnote{Normally, a node is in its own
  transmission range. In that case the received {\chnk} $\dval{ch}$ will be either the {\chnk} $\chunk
  mc$ it is 
  transmitting itself, or \conflict\ in case some other node within transmission range is
  transmitting as well.} \pagebreak
The counter \counter\ is $0$ before a message is sent, and is incremented before
the transmission of each chunk.
So,\ each {\chnk} sent has the form $\chunk{\xi(\dexp{ms})}{\xi(\counter){+}1}$.
To ease readability we abbreviate ${\xi(\counter){+}1}$ by $\cinc$.
In case the (already incremented) counter $\cinc$ is strictly smaller than the number of {\chnk}s needed to send $\xi(\dexp{ms})$,
another \textbf{transmit}-action is needed (Rule 1); if the last fragment has been sent ($\cinc=\dur{\xi(\dexp{ms})}$) the process can continue to act as $\p$ (Rule 2). 

The actions $\obtain{d}{\dval{dest}}$ and $\deliver{\dval{d}}$ are instantaneous and model the submission of data $d$ from the \application layer, destined for $\dval{dest}$,
and the delivery of data $d$ to the \application layer, respectively.
The process $\obtain{d}{\dval{dest}}.P$ has also the possibility to wait, namely if no \application layer instruction arrives.

Rule 6 defines a rule for assignment in a straightforward fashion; only 
the valuation of the variable \keyw{var} is updated.

In Rules\ 7 and 8, which define recursion, 
$\xi_{|\READONLY}[\keyw{var}_i:=\xi(\dexp{exp}_i)]_{i=1}^n$ 
is the valuation that
\emph{only} assigns the values $\xi(\dexp{exp}_i)$ to the variables $\keyw{var}_i$,
for $i=1,\ldots,n$, and maintains the values of the variables $\now$, $\rfr$ and $\counter$.
These rules state that a defined process $X$ has the same
transitions as the body $p$ of its defining equation.
In case of a $\w$-transition, the sequential process does not progress, and accordingly the
recursion is not yet unfolded.

Most transition rules so far feature statements of the form $\xi(\dexp{exp})$ where \dexp{exp} is
  a data expression. The application of the rule depends on
  $\xi(\dexp{exp})$ being defined.
Rule 9 covers all cases where the above rules cannot  be applied since at least one data expression in an action $\alpha$ is not defined.
A state $\xi,\p$ is \emph{unvalued}, denoted by $\xi(p)\undefined$, if $\p$ has the form
$\transmit{\dexp{ms}}.\p$,
$\deliver{\keyw{data}}.\p$,
$\assignment{\keyw{var}:=\dexp{exp}}\p\,$ or
$\,X(\dexp{exp}_1,\ldots,\dexp{exp}_n)$ with
either $\xi(\dexp{ms})$ or $\xi(\keyw{data})$ or
$\xi(\dexp{exp})$ or some $\xi(\dexp{exp}_i)$ undefined. From such a state the process can
merely wait.

A process $\p+\q$ can wait \emph{only} if both $\p$ and $\q$ can do the same;
if either $\p$ or $\q$ can achieve `proper' progress, the choice process $\p+\q$
always chooses progress over waiting.
A simple induction shows that
if $\xi,P\ar{\w(\dval{ch})}\zeta, P'$ and $\xi,Q\ar{\w(\dval{ch})}\zeta', Q'$
then $P=P'$, $Q=Q'$ and $\zeta=\zeta'$.

The first rule of (12),  describing the semantics of guards $\cond{\varphi}$,  is taken from \awn.
Here \plat{$\xi
  \stackrel{\varphi}{\rightarrow}\zeta$} says that $\zeta$ is an
extension of $\xi$, i.e.\ a valuation that agrees with $\xi$ on all
variables on which $\xi$ is defined, and evaluates other variables
occurring free in $\varphi$, such that the formula $\varphi$ holds
under $\zeta$. All variables not free in $\varphi$ and not evaluated by $\xi$ 
are also not evaluated by $\zeta$.
Its negation \plat{$\xi\nar{\varphi}$} says that no such extension exists, and
thus, that $\varphi$ is false in the current state, no matter how we interpret
the variables whose values are still undefined.
If that is the case, the process $[\varphi]p$ will idle by performing the
action $\w(\dval{ch})$.

\subsection{A Language for Node Expressions}
\newcommand{\mymid}{\!\mid\!}
We model network nodes in the context of a (wireless) network by
\emph{node expressions} of the form\\[-2ex]
\centerline{$\nexpr{\dval{id}}{(\xi,\SP)}{R}\ .$}\vspace{1ex}
Here $\dval{id} \in \tIP$ is the \emph{address} of the node, $\SP$ is a
sequential process expression with a valuation $\xi$, and $R\in\pow(\tIP)$ is 
the \emph{range} of the node, defined as the set of nodes within transmission
range of $\dval{id}$. Unlike \awn, the process algebra does not offer a parallel operator 
for combining sequential processes; such an operator is not needed due to the nature of link layer protocols.

In the semantics of this layer it is crucial to handle frame collisions. 
The idea is that all {\chnk}s sent are recorded, together with the respective recipient. 
In case a node receives more than one chunk at a time, a conflict is raised, 
as it is\linebreak[3] impossible to send two or more messages via the same medium at the same time. 

The formal semantics for node expressions, presented in \Tab{sos node}, uses
transition labels $\transmission{\tset{T}}{\tset{R}}$,
$\colonact{\dval{id}}{\deliver{\dval{d}}}$,
$\colonact{\dval{id}}{\obtain{d}{\dval{id\/}'}}$,
$\connect{\dval{id}}{\dval{id\/}'}$, $\disconnect{\dval{id}}{\dval{id\/}'}$ and $\tau$, 
with partial functions $\tset{T},\tset{R}\mathop:\tIP\rightharpoonup\tCHUNK$,
$\dval{id}, \dval{id\/}'\in\tIP$, and $d\in\tDATA$.

\begin{table}[t]
\vspace{-12pt}
\small
\[\begin{array}{@{}c@{}}
\displaystyle
\frac{P \ar{\w(\medIdle)} P'}
{\rule[11pt]{0pt}{1pt}
	\nexpr{\dval{id}}{P}{R}
	\ar{\transmission{\emptyset}{\emptyset}}
	\nexpr{\dval{id}}{P'\!}{R}}
\qquad
\displaystyle
\frac
{P \ar{\transmitAct{m}{c}{\medIdle}} P'}
{\rule[11pt]{0pt}{1pt}
	\nexpr{\dval{id}}{P}{R}
	\ar{\transmission{\{(r,\chunk{m}{c})\mymid r{\in} R\}}{\emptyset}}
	\nexpr{\dval{id}}{P'\!}{R}}
\\[16pt]
\displaystyle
\frac
{P \ar{\w(\dval{ch})}\quad\scriptstyle(\dval{ch}{\neq}\medIdle)}
{\rule[11pt]{0pt}{1pt}
	\nexpr{\dval{id}}{P}{R}
	\ar{\transmission{\emptyset}{\{(\dval{id},\dval{ch})\}}}
	\nexpr{\dval{id}}{P'\!}{R}}
\quad
\frac
{P \ar{\transmitAct{m}{c}{\dval{ch}}} P'\quad\scriptstyle(\dval{ch}{\neq}\medIdle)}
{\rule[11pt]{0pt}{1pt}
	\nexpr{\dval{id}}{P}{R}
	\ar{\transmission{\{(r,\chunk{m}{c})\mymid r{\in} R\}}{\{(\dval{id},\dval{ch})\}}}
	\nexpr{\dval{id}}{P'\!}{R}}
\\[16pt]
\displaystyle
\frac
{P \ar{\deliver{\dval{d}}} P'}
{\rule[11pt]{0pt}{1pt}
	\nexpr{\dval{id}}{P}{\!R}
	\mathbin{\ar{\colonact{\dval{id}}{\deliver{\dval{d}}}}}
	\nexpr{\dval{id}}{P'\!}{\!R}}
\quad
\frac
{P \ar{\obtain{d}{\dval{dest}}} P'}
{\rule[11pt]{0pt}{1pt}
	\nexpr{\dval{id}}{P}{\!R}
	\mathbin{\ar{\colonact{\dval{id}}{\obtain{d}{\dval{dest}}}}}
	\nexpr{\dval{id}}{P'\!}{\!R}}
\quad
\displaystyle
\frac
{P \ar{\tau} P'}
{\nexpr{\dval{id}}{P}{\!R}
	\mathbin{\ar{\tau}}
	\nexpr{\dval{id}}{P'\!}{\!R}}
\\[16pt]
\displaystyle
\nexpr{\dval{id}}{P}{R}
\ar{\connect{\dval{id}}{\dval{id\/}'}}
\nexpr{\dval{id}}{P}{R\cup\{\dval{id\/}'\}}
\qquad
\nexpr{\dval{id}}{P}{R}
\ar{\disconnect{\dval{id}}{\dval{id\/}'}}
\nexpr{\dval{id}}{P}{R-\{\dval{id\/}'\}}
\\[6pt]
\displaystyle
\nexpr{\dval{id}}{P}{R}
\ar{\connect{\dval{id\/}'}{\dval{id}}}
\nexpr{\dval{id}}{P}{R\cup\{\dval{id\/}'\}}
\qquad
\nexpr{\dval{id}}{P}{R}
\ar{\disconnect{\dval{id\/}'}{\dval{id}}}
\nexpr{\dval{id}}{P}{R-\{\dval{id\/}'\}}
\\[4pt]\displaystyle
  \frac{\dval{id} \not\in \{\dval{id\/}'\!,\dval{id\/}''\}}
  {\rule[13pt]{0pt}{1pt}
   \nexpr{\dval{id}}{P}{R} \ar{\textbf{connect}(\dval{id\/}'\!,\dval{id\/}'')} \nexpr{\dval{id}}{P}{R}}
\qquad
  \frac{\dval{id} \not\in \{\dval{id\/}'\!,\dval{id\/}''\}}
  {\rule[13pt]{0pt}{1pt}
  \nexpr{\dval{id}}{P}{R} \ar{\textbf{disconnect}(\dval{id\/}'\!,\dval{id\/}'')} \nexpr{\dval{id}}{P}{R}}
\end{array}
\]
	\caption{\em Structural operational semantics for node expressions}
	\label{tab:sos node}
	\vspace{-23pt}
\end{table}

All time-consuming actions on process level ($\transmitAct{m}{c}{\!\dval{ch}}$
and $\w(\dval{\!ch})$) are transformed into an action
$\transmission{\tset{T}}{\tset{R}}$ on node level:
the first argument $\tset{T}$ maps
$\dval{dest}$ to $\chunk{m}{c}$ if and only if the {\chnk} $\chunk{m}{c}$ is transmitted to $\dval{dest}$.
The second argument $\tset{R}$
maps $\dval{id}$ to $\chunk{m}{c}$ if and only if the {\chnk} $\chunk{m}{c}$
is received on process level at node $\dval{id}$. For the sos-rules of \Tab{sos node} we use 
the set-theoretic presentation of partial functions.
The two rules for $\textbf{wait}$ set $\tset{T}:=\emptyset$, as no {\chnk}s are transmitted; the
rules for $\textbf{transmit}$ allow a transmitted {\chnk} \chunk{m}{c} to travel to all nodes within transmission
range:  $\tset{T}:=\{(r,\chunk{m}{c})\mymid r\in R\}$.
In case that during the transmission or waiting no {\chnk} is received ($\dval{ch}=\medIdle$)
we set $\tset{R}=\emptyset$; otherwise $\tset{R}=\{(\dval{id},\dval{ch})\}$,
indicating that {\chnk} $\dval{ch}$ is received by node~$\dval{id}$.

The actions $\colonact{\dval{id}}{\obtain{d}{\dval{dest}}}$ and $\colonact{\dval{id}}{\deliver{\dval{d}}}$ as well as the internal 
actions $\tau$
are simply inherited by node expressions from the processes that run
on these nodes.

$\!$The remaining rules of \Tab{sos node} model the mobility aspect of wireless networks; 
the rules are taken straight from \awn~\cite{ESOP12,TR13}.
We allow actions
$\textbf{connect}(\dval{id},\dval{id\/}')$ and
$\textbf{disconnect}(\dval{id},\dval{id\/}')$ for
$\dval{id},\dval{id\/}'\mathop\in \tIP$ modelling a change in network
topology.  These actions can be thought of as occurring nondeterministically, or as
actions instigated by the environment of the modelled network
protocol.  In this formalisation node $\dval{id\/}'$ is in
the range of node $\dval{id}$, meaning that $\dval{id\/}'$ can receive
messages sent by $\dval{id}$, if and only if $\dval{id}$ is in the
range of $\dval{id\/}'$. To break this symmetry, one just skips the last
four rules of Table~\ref{tab:sos node} and replaces the synchronisation
rules for \textbf{connect} and \textbf{disconnect} in
Table~\ref{tab:sos network} by interleaving rules (like the ones for
\textbf{deliver}, \textbf{newpkt} and $\tau$)~\cite{ESOP12}.
For some applications a wired or non-mobile network need to be considered. 
In such cases the last six rules of \Tab{sos node} are dropped.

Whether a node $\nexpr{\dval{id}}{\p}{R}$ receives its own transmissions 
depends on whether $\dval{id}\in R$. Only if $\dval{id}\in R$ our process algebra 
will disallow the transmission from and to a single node $\dval{id}$ at the same time, 
yielding a \conflict.

\subsection{A Language for Networks}
A \emph{partial network} is modelled by a \emph{parallel composition}
$\|$ of node expressions, one for every node in the network. A
\emph{complete network} is a partial network within an
\emph{encapsulation operator} $[\_]$, which limits the communication between
network nodes and the outside world to the receipt and delivery of data packets
to and from the \application layer.

The syntax of networks is described by the following grammar:\\[2mm]
\centerline{$
N ::= [M_{T}^{T}] \qquad
M_{S_1\dcup S_2}^T ::= M_{S_1}^T \| M_{S_2}^T \qquad
M_{\{\dval{id}\}}^T ::= \nexpr{\dval{id}}{(\xi, \SP)}{R}\ ,$}\\[2mm]
with $\{\dval{id}\}\cup R \subseteq T \subseteq \tIP$.
Here $M_S^T$ models a partial network describing the behaviour of all nodes $\dval{id}\in S$.
The set $T$ contains the identifiers of all nodes that are part of the complete network. 
This grammar guarantees that node identifiers of node expressions---the first component of $\nexpr{\dval{id}}{P}{R}$---are unique.

The operational semantics of network expressions is given in \Tab{sos network}.
Internal actions $\tau$ as well as the actions
$\colonact{\dval{id}}{\deliver{\dval{d}}}$ and	
$\colonact{\dval{id}}{\obtain{d}{\!\dval{id}}}$
are interleaved in the parallel composition of
nodes that makes up a network, and then lifted to encapsulated networks (Line 1 of \Tab{sos network}).

\begin{table}[t]
\vspace{-10pt}
\[\begin{array}{@{}c@{}}
\displaystyle
\frac{M \ar{a} M'}{M \| N \ar{a} M' \| N}
\qquad\!\!\!
\frac{N \ar{a} N'}{M \| N \ar{a} M \| N'}
\qquad\!\!\!
\frac{M \ar{a} M'}{[M] \ar{a} [M']}
\qquad\!\!\mbox{\scriptsize
	$
	\left(\!\forall a\!\mathbin\in\!\left\{\begin{array}{@{}l@{}}
	\tau,\colonact{\dval{id}}{\deliver{\dval{d}}},\\
	\colonact{\dval{id}}{\obtain{d}{\dval{id}}},\\
	\end{array}\right\}\!\right)$}
\\[16pt]\displaystyle
  \frac{M \ar{a} M' \quad
        N \ar{a} N'}
  {M \| N \ar{a} M' \| N'}
\qquad\!\!\!
\frac{M \ar{a} M'}{[M] \ar{a} [M']}
\qquad\!\!\mbox{\scriptsize
	$
	\left(\!\forall a\!\mathbin\in\!\left\{\begin{array}{@{}l@{}}
	\connect{\dval{id}}{\dval{id\/}'},\\
	\disconnect{\dval{id}}{\dval{id\/}'}\!
	\end{array}\right\}\!\right)$}
\\[16pt]
\displaystyle
\frac
{M
	\ar{\transmission{\tset{T}_1}{\tset{R}_1}}
	M' \quad
	N
	\ar{\transmission{\tset{T}_2}{\tset{R}_2}}
	N'
}
{\rule[11pt]{0pt}{1pt}
	M \| N
	\ar{\transmission{\tset{T}_1\uplus\tset{T}_2}{\tset{R}_1\uplus\tset{R}_2}}
	M' \| N'
	}
\qquad
\displaystyle
\frac{M \ar{\transmission{\tset{R}}{\tset{R}}} M'}{
	[M] \ar{{\bf tick}} [M']}
\end{array}
\]
\caption{\em Structural operational semantics for network expressions}
	\label{tab:sos network}
	\vspace{-25pt}
\end{table}

\newcommand{\dom}[1]{\texttt{dom}(#1)}
Actions \textbf{traffic} and (\textbf{dis})\textbf{connect} are synchronised. 
The rule for synchronising the action \transmissionID\ (Line 3), the only
action that consumes time on the network layer, uses the union $\uplus$ of partial functions. It is formally defined as
\[
(\tset{R}_1\uplus\tset{R}_2)(id):=
\left\{
\begin{array}{l@{\quad}l}
\conflict&\mbox{if $id\in\dom{\tset{R}_1}\cap\dom{\tset{R}_2}$}\\
\tset{R}_1(id)&\mbox{if $id\in\dom{\tset{R}_1}-\dom{\tset{R}_2}$}\\
\tset{R}_2(id)&\mbox{if $id\in\dom{\tset{R}_2}-\dom{\tset{R}_1}$}\ .
\end{array}
\right.
\]
The synchronisation of the sets $\tset{R}_i$ and $\tset{T}_i$ has the following intuition:
if a node identifier $\dval{id}\in\tIP$ is in both $\dom{\tset{T}_1}$ and
\dom{$\tset{T}_2$} then there exist two nodes that transmit to
node $\dval{id}$ at the same time, and therefore a frame collision occurs. 
In our algebra this is modelled by the special {\chnk} $\conflict$. The sos rules of Tables~\ref{tab:sos node} and \ref{tab:sos network}
guarantee that there cannot be collisions within the set of received {\chnk}s $\tset{R}$.
The reason is that each node merely contributes to $\tset{R}$ a {\chnk} for itself; 
it can be the {\chnk} $\conflict$ though.
Therefore we could have written $\tset{R}_1\cup\tset{R}_2$ instead of $\tset{R}_1\uplus\tset{R}_2$
in the sixth rule of \Tab{sos network}.

The last rule propagates a $\transmission{\tset{T}}{\tset{R}}$-action of a partial network
$M$ to a complete network $[M]$. By then $\tset{T}$ consists of all {\chnk}s (after collision
detection) that are being transmitted by any member in the
network, and $\tset{R}$ consists of all  {\chnk}s that are
received. The condition $\tset{R}=\tset{T}$ determines the content of the messages in $\tset{R}$.
The $\transmission{\tset{T}}{\tset{R}}$-actions become internal at this level,
as they cannot be steered by the outside world; all that is left is a time-step \textbf{tick}.

\subsection{Results on the Process Algebra}

As for the process algebra T-AWN \cite{ESOP16}, but with a slightly simplified proof, one can show
that our processes have no \emph{time deadlocks}:

\begin{theorem}
A complete network $N$ in our process algebra always admits a transition, independently of the outside environment, 
i.e.\ $\forall N, \exists a$ such that $N \ar{a}$ and
$a \not \in \{\textbf{connect}(\dval{id},\dval{id\/}'),\textbf{disconnect}(\dval{id},\dval{id\/}'),\dval{id}\!:\!\newpkt{\dval{d}}{\dval{dest}}\}$.\\
More precisely, either $N \ar{\bf tick}$, or $N \ar{\colonact{\dval{id}}{\deliver{\dval{d}}}}$ or $N \ar{\tau}$.
\end{theorem}

The following results (statements and proofs) are very similar to the results about the process algebra \awn, as presented in \cite{TR13}.
A rich body of foundational meta theory of process algebra allows the transfer of the results to our
setting, without too much overhead work.

Identical to \awn and its timed version {\sc T-AWN}, 
our process algebra admits a translation into one without data structures
(although we cannot describe the target algebra without using data structures).
The idea is to replace any variable by all
possible values it can take. The target algebra differs from the original only on the level of sequential processes; the subsequent layers are unchanged. The construction closely follows the one given in the appendix of \cite{ESOP16}.
The inductive definition contains the rules 

$	\T_\xi(\deliver{\dexp{data}}.\p)=\deliver{\xi(\dexp{data})}.\T_\xi(\p)$ and

$\T_\xi(\assignment{\keyw{var}:=\dexp{exp}}\p)=\tau.\T_{\multiAssign{\keyw{var}:=\xi(\dexp{exp})}}(\p)$.

\noindent
Most other rules require extra operators that keep track of the passage of time and the
  evolution of other internal variables.
The resulting process algebra has a structural operational semantics in the
(infinitary)
\emph{de Simone} format, generating the same transition system---up to strong
bisimilarity, $\bis$ ---as the original. It follows that $\bis\,$, and many other
semantic equivalences, are congruences on our language \cite{dS85}.

\begin{theorem}\label{thm:1}
Strong bisimilarity is a congruence for all operators of our language.
\end{theorem}

\noindent
This is a deep result that usually takes many pages to establish (e.g.~\cite{SRS10}).
Here we get it directly from the existing theory on structural
operational semantics, as a result of carefully designing our
language within the disciplined framework described by de Simone~\cite{dS85}.
\qed

\begin{theorem}\label{thm:2}
The operator $\|$ is associative and commutative, up to $\bis$\,.
\end{theorem}

\begin{proof}
The operational rules for this operator fits a format presented in \cite{CMR08},
guaranteeing associativity up to~$\bis$.
The {\emph ASSOC-de Simone format} of \cite{CMR08} applies to all 
transition system specifications (TSSs) in de
Simone format, and allows $7$ different types of rules (named $1$--$7$) for the operators in question.
Our TSS is in de Simone format; the four rules for $\|$ of
\Tab{sos network} are of types $1$, $2$ and $7$, respectively.
To be precise, it has rules $1_a$ and $2_a$ for $a\in\{\tau$,
$\colonact{\dval{id}}{\deliver{\dval{d}}}$,
$\colonact{\dval{id}}{\obtain{d}{\dval{dest}}}\}$, rules $7_{(a,b)}$ for
\[(a,b)\mathbin\in\{(\transmission{\tset{T}_1}{\tset{R}_1},\transmission{\tset{T}_2}{\tset{R}_2})\mid \tset{R}_1,\tset{R}_2,\tset{T}_1,\tset{T}_2\mathop\in\tIP\rightharpoonup\tCHUNK\}\]
and rules $7_{(c,c)}$ for $c\in\{\connect{\dval{id}}{\dval{id\/}'},\disconnect{\dval{id}}{\dval{id\/}'} \mid \dval{id},\dval{id\/}'\in\tIP\}$.
Moreover, the partial \emph{communication function}
$\gamma:\act\times\act\rightharpoonup \act$ is given by
$\gamma(\transmission{\tset{T}_1}{\tset{R}_1},\transmission{\tset{T}_2}{\tset{R}_2})=\transmission{\tset{T}_1\uplus\tset{T}_2}{\tset{R}_1\uplus\tset{R}_2}$
and $\gamma(c,c)=c$.
The main result of \cite{CMR08} is that an operator is guaranteed to
be associative, provided that $\gamma$ is associative and six
conditions are fulfilled. In the absence of rules of types 3, 4, 5
and 6, five of these conditions are trivially fulfilled, and the
remaining one reduces to\\[1mm]
\centerline{$
7_{(a,b)} \Rightarrow (1_a \Leftrightarrow 2_b)
           \wedge (2_a \Leftrightarrow 2_{\gamma(a,b)})
           \wedge (1_b \Leftrightarrow 1_{\gamma(a,b)})\ .$}\\[2mm]
Here $1_a$ says that rule $1_a$ is present, etc.
This condition is trivially met for $\|$ as
there neither exists a rule of the form $1_{\transmission{\tset{T}\!}{\tset{R}}}$ nor of the form 
$2_{\transmission{\tset{T}\!}{\tset{R}}}$, or $1_c$, $2_c$ with $c$ as above. As on
\textbf{traffic} actions~$\gamma$ is basically the union of partial functions
($\uplus$), where a collision in domains is indicated by an error $\conflict$, it is
straightforward to prove associativity of~$\gamma$.

Commutativity of $\|$ follows by symmetry of the sos rules.
\qed
\end{proof}

%%%%%%%%%%%%%%%%%%%%%%%%%%%%%%%%%%
\section{An Algebra for Link Layer Protocols}\label{sec:pa_2}
%%%%%%%%%%%%%%%%%%%%%%%%%%%%%%%%%%
We now introduce \all, the \emph{Algebra for Link Layer protocols}. It is obtained from the process
algebra presented in the previous section by the addition of a probabilistic choice operator $\bigoplus_{0}^n$.
As a consequence, the semantics of the algebra is no longer a labelled transition system, but a 
\emph{probabilistic labelled transition system} (pLTS)~\cite{DGHMZ07}.
This is a triple $(S, \act, \rightarrow)$, where
\begin{enumerate}[(i)]
	\item $S$ is a set of states
	\item $\act$ is a set of actions
	\item ${\rightarrow} \subseteq S \times \act \times \dist{S}$, where $\dist{S}$ is the set
          of all (discrete) probability distributions over $S$: functions
          $\Delta:S\rightarrow[0,1]$ with $\sum_{s\in S}\Delta(s)=1$.
\end{enumerate}
As with LTSs, we usually write $s \ar{\alpha} \Delta$ instead of $(s,\alpha,\Delta) \in {\rightarrow}$.
The \emph{point distribution} $\delta_s$, for $s\in S$, is the distribution with $\delta_s(s)=1$.
We simply write \plat{$s \ar{\alpha} t$} for \plat{$s \ar{\alpha} \delta_t$}.
An LTS may be viewed as a degenerate pLTS, in which only point distributions occur.
For a uniform distribution over $s_0,\dots,s_n\in S$ we write $\unidist_{i=0}^n s_i$.
The pLTS associated to \all takes $S$ to be the disjoint union of the pairs $\xi,\p$, with $\p$ a
sequential process expression, and the network expressions. $\act$ is the collection of transition
labels, and $\rightarrow$ consists of the transitions derivable from the structural operational
semantics of the language.

Rules (1)--(6), (9), (11) and (12) of \Tab{sos sequential} are adopted to \all unchanged,
whereas in Rules (7), (8) and (10) the state $\zeta,\p'$ (or
$\zeta,\q'$) is replaced by an arbitrary distribution $\Delta$.
Add to those the following rule for the probabilistic choice operator:
\vspace{-1.5mm}
\[\begin{array}{@{}r@{~\hspace{-.5pt}}c@{~\hspace{-.5pt}}l@{\quad}r}
\displaystyle
	\xi,\PC{\keyw{i}}{0}{n}\p\,
	\ar{\tau}\
	\unidist_{i=0}^{\xi(\dval{n})}~\multiAssign{\keyw{i}:=i},\p
\end{array}\]
Here the data variable $i$ may occur in $P$. 
The rules of Tables \ref{tab:sos node} and \ref{tab:sos network} are
adapted to \all unchanged, except that $P'$, $M'$ and $N'$ are now replaced by arbitrary distributions over
sequential processes and network expressions, respectively. Here we adapt the convention that a
unary or binary operation on states lifts to distributions in the standard manner.
For example, if $\Delta$ is a distribution over sequential processes, $\dval{id}\in\tIP$ and
$R\subseteq\tIP$, then $\nexpr{\dval{id}}{\Delta}{R}$ describes the distribution over node
expressions that only has probability mass on nodes with address $\dval{id}$ and range $R$, and for
which the probability of $\nexpr{\dval{id}}{P}{R}$ is $\Delta(P)$.
Likewise, if $\Delta$ and $\Theta$ are distributions over network expressions, then
$\Delta\|\Theta$ is the distribution over network expressions of the form $M\|N$,
where $(\Delta\|\Theta)(M\|N) = \Delta(M)\cdot\Theta(N)$.

\section{\hspace{-.5pt}Formalising Liveness Properties of Link Layer Protocols}\label{sec:liveness}

Link layer protocols communicate with the \application layer through the actions
$\colonact{\dval{id}}{\obtain{d}{\dval{dest}}}$ and
$\colonact{\dval{id}}{\deliver{\dval{d}}}$.
The typical liveness property expected of a link layer protocol is that if 
the \application layer at node $\dval{id}$ injects
a data packet $d$ for delivery at destination $\dval{dest}$
then this packet is delivered eventually. In terms of our process algebra, this says that every
execution of the action $\colonact{\dval{id}}{\obtain{d}{\dval{dest}}}$ ought to be followed by 
the action $\colonact{\dval{dest}}{\deliver{\dval{d}}}$.
This property can be formalised in Linear-time Temporal Logic \cite{Pn77} as
\begin{equation}\label{L1}
  {\bf G}\big(\colonact{\dval{id}}{\obtain{d}{\dval{dest}}} \Rightarrow
  {\bf F}(\colonact{\dval{dest}}{\deliver{\dval{d}}}) \big)
\end{equation}
for any $\dval{id},\dval{dest}\mathbin\in\tIP$ and $d\mathbin\in\tDATA$.
This formula has the shape \plat{$\mathbf{G} \big(\phi^{\it pre} \Rightarrow \mathbf{F}\phi^{\it post}\big)$},
and is called an \emph{eventuality property} in \cite{Pn77}. It says that whenever we reach a state
in which the precondition $\phi^{\it pre}$ is satisfied, this state will surely be followed by a
state were the postcondition $\phi^{\it post}$ holds. In \cite{DV95,TR13} it is explained how
action occurrences can be seen or encoded as state-based conditions.
Here we will not define how to interpret general LTL-formula in pLTSs,
but below we do this for eventuality properties with specific choices of $\phi^{\it pre}$ and  $\phi^{\it post}$.

Formula (\ref{L1}) is too strong and does not hold in general: in case the nodes $\dval{id}$ and
$\dval{dest}$ are not within transmission range of each other, the delivery of messages from
$\dval{id}$ to $\dval{dest}$ is doomed to fail. We need to postulate two side conditions to make this
liveness property plausible. Firstly, when the request to deliver the message comes in, 
$\dval{id}$ needs to be connected to $\dval{dest}$.\pagebreak[3]
We introduce the predicate $\textbf{cntd}(\dval{id},\dval{dest})$ to express this, and hence take
$\phi^{\it pre}$ to be $\textbf{cntd}(\dval{id},\dval{dest}) \wedge \colonact{\dval{id}}{\obtain{d}{\dval{dest}}}$.
Secondly, we assume that the link between $\dval{id}$ and $\dval{dest}$ does not break until the message is delivered.
As remarked in \cite{TR13}, such a side condition can be formalised by taking $\phi^{\it post}$ to be
$\colonact{\dval{dest}}{\deliver{\dval{d}}} \vee \textbf{disconnect}(\dval{id},\dval{dest})$.
Thus the liveness property we are after is
\begin{equation}\label{L2}
\begin{array}{c}
  {\bf G}\big(\textbf{cntd}(\dval{id},\dval{dest}) \wedge
  \colonact{\dval{id}}{\obtain{d}{\dval{dest}}} \Rightarrow \\
  {\bf F}(\colonact{\dval{dest}}{\deliver{\dval{d}}} \vee \textbf{disconnect}(\dval{id},\dval{dest}) \vee \textbf{disconnect}(\dval{dest},\dval{id})) \big)
\end{array}
\end{equation}

We now define the validity of eventuality properties
\plat{$\mathbf{G} \big(\phi^{\it pre} \Rightarrow \mathbf{F}\phi^{\it post}\big)$}. Here $\phi^{\it pre}$ and $\phi^{\it post}$ denote sets of transitions and actions, respectively,
and hold if one of the transitions or actions in the set occurs. In (\ref{L2}), 
$\phi^{\it pre}$ denotes the transitions with label
$\colonact{\dval{id}}{\obtain{d}{\dval{dest}}}$
that occur when the side condition $\textbf{cntd}(\dval{id},\dval{dest})$ is met,
whereas
$\phi^{\it post} = \{\colonact{\dval{dest}}{\deliver{\dval{d}}},\textbf{disconnect}(\dval{id},\dval{dest}),\linebreak[2]\textbf{disconnect}(\dval{dest},\dval{id})\}$ 
is a set of actions.
\pagebreak[3]

A \emph{path} in a pLTS $(S,\act,\rightarrow)$ is an alternating sequence
$s_0,\alpha_1,s_1,\alpha_2,\dots$ of states and actions, starting with a state and either being infinite
or ending with a state, such that there is a transition $s_i\ar{\alpha_{i+1}}\Delta_{i+1}$ with
$\Delta_{i+1}(s_{i+1})>0$ for each $i$.\linebreak[3] The path is \emph{rooted} if it starts with a
state marked as `initial', and \emph{complete} if either it is infinite, or there is no transition 
starting from its last state.
A state or transition is \emph{reachable} if it occurs in a rooted path.

In a pLTS with an initial state, an eventually formula
\plat{$\mathbf{G} \big(\phi^{\it pre} \Rightarrow \mathbf{F}\phi^{\it post}\big)$},
with $\phi^{\it pre}$ and $\phi^{\it post}$ denoting sets of transitions and actions, \emph{holds outright}
if all complete paths starting with a reachable transition from $\phi^{\it pre}$
contain a transition with a label from $\phi^{\it post}$.

Definitions 3 and 5 in \cite{DGMZ07} define the set of probabilities that a pLTS with an initial
state will ever execute the action $\omega$. One obtains a set of probabilities rather than a single
probability due to the possibility of nondeterministic choice. This definition generalises to \emph{sets} of
actions $\phi^{\it post}$ (seen as disjunctions) by first renaming all actions in such a set into $\omega$.
It also generalises trivially to pLTSs with an \emph{initial transition}.
For $t$ a transition in a pLTS, let $\mathit{Prob}(t,\phi^{\it post})$ be the infimum of the set
of probabilities that the pLTS in which $t$ is taken to be the initial transition will ever execute $\phi^{\it post}$.
Now in a pLTS with an initial state, an eventually formula
\plat{$\mathbf{G} \big(\phi^{\it pre} \Rightarrow \mathbf{F}\phi^{\it post}\big)$}
\emph{holds with probability at least $p$} if for all reachable transitions $t$ in
$\phi^{\it pre}$ we have $\mathit{Prob}(t,\phi^{\it post})\geq p$.

Possible correctness criteria for link layer protocols are that the liveness property (\ref{L2})
either holds outright, holds with probability 1, or at least holds with probability $p$ for a
sufficiently high value of $p$.

Sometimes we are content to establish that (\ref{L2}) holds under the additional assumptions that the
network is stable until our packet is delivered, meaning that no links between any nodes are broken
or established, and/or that the \application layer refrains from injecting more packets.
This is modelled by taking
\vspace{-1mm}
\begin{equation}\label{L3post}
\hspace{-2pt}\phi^{\it post} = \{\begin{array}{@{}l@{}}\colonact{\dval{dest}}{\deliver{\dval{d}}},
                                       \textbf{disconnect}(*,*),
                                       \textbf{connect}(*,*),
                                       \obtain{*}{*}
\end{array}\}.
\end{equation}
We will refer to this version of (\ref{L2}) as the \emph{weak packet delivery} property.
\emph{Packet delivery} is the strengthening without $\textbf{newpkt}(*,*)$  in (\ref{L3post}),
i.e.\ not assuming that the \application layer refrains from injecting more packets.

%%%%%%%%%%%%%%%%%%%%%%%%%%%%%%%%%%
\section{Modelling and Analysing the CSMA/CA Protocol}\label{sec:csma}
%%%%%%%%%%%%%%%%%%%%%%%%%%%%%%%%%%

In this section we model two versions of the CSMA/CA protocol, using the process algebra \all.
Moreover, we briefly discuss some results we obtained while analysing these protocols. 

The \emph{Carrier-Sense Multiple Access} (CSMA) protocol is a media access control (MAC) protocol in which a node verifies the absence of other traffic before transmitting on a shared transmission medium.
If a carrier is sensed, the node waits for the transmission in progress to end before initiating its own transmission.
Using CSMA, multiple nodes may, in turn, send and receive on the same medium.
 Transmissions by one node are generally received by all other nodes connected to the medium.

The CSMA protocol with Collision Avoidance (CSMA/CA)~\cite{IEEE80211,IEEE80215}%
\footnote{The primary medium access control (MAC) technique of IEEE 802.11~\cite{IEEE80211} is called \emph{distributed coordination function} (DCF), which is a CSMA/CA protocol.}  improves the performance of CSMA. 
If the transmission medium is sensed busy before transmission then the transmission is deferred for a \emph{random} time interval. This interval reduces the likelihood that two or more nodes waiting to transmit will simultaneously begin transmission upon termination of the detected transmission. 
CSMA/CA is used, for example, in Wi-Fi.

It is well known that CSMA/CA suffers from the \emph{hidden station problem} (see \autoref{sec:Hidden}).
To overcome this problem, CSMA/CA is often supplemented by the 
request-to-send/clear-to-send (\rts/\cts) handshaking \cite{IEEE80211}.
This mechanism is known as the IEEE 802.11 \rts/\cts exchange, or \emph{virtual carrier
  sensing}. While this extension reduces the amount of collisions, wireless 802.11 implementations
do not typically implement \rts/\cts for all transmissions because the transmission overhead is too great for small data transfers. 

We use the process algebra \all to model both the CSMA/CA without and with virtual carrier sensing.

\subsection{A Formal Model for CSMA/CA}\label{sec:csma-default}
Our formal specification of CSMA/CA consists of four short processes written in \all. 
It is precise and free of ambiguities---one of the many advantages formal methods provide, 
in contrast to specifications written in English prose.

The syntax of \all is intended to look like pseudo code, and it is our belief that the
specification can easily be read and understood by software engineers, who may or may not have
experience with process algebra.

As the underlying data structure of our model is straightforward, we do not present it explicitly, but introduce it while describing the different processes.

%%%%%%%%%%%%%%%%
%% Process 1
The basic process \CSMA, depicted in \Proc{p_csma}, is the protocol's entry point.

  \vspace*{-2.65ex}
  \algsetup{linenodelimiter=.,linenosize=\tiny}
  \begin{algorithm}[H]
    {\footnotesize
      \caption{\small The Basic Routine}
      \label{pro:p_csma}
      \begin{algorithmic}[1]
        	% !TEX root = ../all.tex
%%%%%%%%%%%%%%%%%%%
\DEFPROCESS{\CSMA}{\myip}
  \STATE\obtain{\data}{\dest}.
  \procInit{\myip}{0}{\dataframe{\data}{\myip}{\dest}}         		 						\label{pro:csma:line1}
  \STATE$+$ 
    \cond{$\new(\dataframe{\data}{\source}{\myip})$}\ \deliver{\data}\ .	\label{pro:csma:line2}
    \IFempty
      \PAR																										\label{pro:csma:line3}
	\UPD{\timeout:=\now+\sifs}			        														\label{pro:csma:line4}
	\cond{$\now \geq \timeout$}			        														\label{pro:csma:line5}
	\STATE\transmit{\ackframe{\source}}\ .\ \procCSMA{\myip}						\label{pro:csma:line6}
      \ENDPAR
    \ENDIFii

		\end{algorithmic}
    }%end{footnotesize}
  \end{algorithm}
        \vspace{-2.65ex}

\noindent
This process maintains a single data variable {\myip} in which it stores its own identity.
It waits until either it receives a request from the \application layer to transmit a packet $\data$
to destination $\dest$, or it receives from another node in the network a CSMA
message (data frame) destined for itself.

In case of a newly injected data packet (\Line{csma}{1}), the process $\INIT$ is called;
this process (described below) initiates the sending of the message via the medium. 
When passing the message on to $\INIT$ we use a function $\dataframeID\mathop:\tDATA\times\tIP\times\tIP\rightarrow\tMSG$
that generates a message in a format used by the protocol: next to the header fields (from which we abstract) it contains the injected \data\ as well as the designated receiver $\dest$ and the sender $\myip$---the current node.

In case of  an incoming $\dataframeID$ destined for this node (the third argument carrying the
destination is $\myip$) (\Line{csma}{2})---any other incoming message is ignored by this process---%
the \data\ is handed over to the \application layer ($\deliver{\data}$) followed
by the transmission of an acknowledgement back to the sender of the message ($\source$).
CSMA/CA requires a short period of idling medium before sending the acknowledgement: in
\cite{IEEE80211} this interval is called \emph{short interframe space} (\sifs).
The process waits until the time of the interframe spacing has passed, and then
transmits the acknowledgement.
The acknowledgement sent is not always received by $\source$, e.g.\ due to data collision;
therefore \source\ could send the same message again (see \Proc{p_ackrecv}) 
and \myip\ could deliver the same data to the \application layer  again.

%%%%%%%%%%%%%%%%
%% Process 2

%
  \vspace*{-2.65ex}
  \algsetup{linenodelimiter=.,linenosize=\tiny}
  \begin{algorithm}[H]
    {\footnotesize
      \caption{\small Protocol Initialisation}
      \label{pro:p_init}
      \begin{algorithmic}[1]
        	% !TEX root = ../all.tex
%%%%%%%%%%%%%%%%%%
\DEFPROCESS{\INIT}{\myip\comma\backoffexp\comma\frmvar}
	\STATE\cond{$\backoffexp \leq \maxRetransmit$}\label{pro:init:line1}%
	\IFempty
		\UPD{\cw := \cwmin \times 2^\backoffexp}																		\label{pro:init:line2}
		\STATE\PC{\backoff}{0}{\cw{-}1}\procCCA{\myip}{\backoff}{\backoffexp}{\frmvar}\COMMENT{choose a backoff from $\{0,\dots,\cw{-}1\}$}\label{pro:init:line3}
	\ENDIFii
	\STATE$+$ \cond{$\backoffexp > \maxRetransmit$}\label{pro:init:line4}%
	\IFempty
		\deliverL{\returnCAF}\ .\ \procCSMA{\myip}																	\label{pro:init:line5}
  \ENDIFii
		\end{algorithmic}
    }%end{footnotesize}
  \end{algorithm}
        \vspace{-2.65ex}

The process $\INIT$ (\Proc{p_init}) initiates the sending of a message via the medium. 
Next to the variable {\myip}, which is maintained by all processes, it maintains 
the variable $\backoffexp$ and $\frmvar$:
$\backoffexp$ stores the number of attempts already made to send message $\frmvar$.
When the process is called the first time for a message $\frmvar$ (\Line{csma}{1} of \Proc{p_csma}) the 
value of $\backoffexp$ is $0$.

The constant $\maxRetransmit$ specifies the maximum number of attempts the protocol is allowed to retransmit the same message.
If the limit is not yet reached (\Line{init}{1})
the message $\frmvar$ is sent.
As mentioned above, CSMA/CA defers messages for a \emph{random} time interval to avoid collision. 
The node must start transmission within the contention window $\cw$, a.k.a.\ backoff time.
\cw\ is calculated in \Line{init}{2}; it increases exponentially.%
\footnote{A typical value for $\cwmin$ is $16$; it must satisfy $\cwmin>0$.}
After $\cw$ is determined, the process \CCA\ is called, which performs the actual
\textbf{transmit}-action. In case the maximum number of retransmits is reached (\Line{init}{4}), the
process notifies the \application layer and restarts the protocol, awaiting new instructions
  from the application layer, or a new incoming message.

%%%%%%%%%%%%%%%%
%% Process 3
\Proc{p_cca}  takes care of the actual transmission of $\frmvar$. However, the protocol has a complicated procedure when to send this message.
\pagebreak[3]

  \vspace*{-2.65ex}
  \algsetup{linenodelimiter=.,linenosize=\tiny}
  \begin{algorithm}[H]
    {\footnotesize
      \caption{\small Clear Channel Assessment With Physical Carrier Sense}
      \label{pro:p_cca}
      \begin{algorithmic}[1]
        	% !TEX root = ../all.tex
%%%%%%%%%%%%%%%%%%
\DEFPROCESS{\CCA}{\myip\comma\backoff\comma\backoffexp\comma\frmvar}
  \STATE\cond{$\new(\dataframe{\data}{\source}{\myip})$}\ \deliver{\data}\ .	\label{pro:cca:line1}
  \IFempty
      \PAR							\label{pro:cca:line2}
	\UPD{\timeout:=\now+\sifs}			        \label{pro:cca:line3}
	\cond{$\now \geq \timeout$}			        \label{pro:cca:line3b}
 	\STATE\transmit{\ackframe{\source}}\ .\ \procCCA{\myip}{\backoff}{\backoffexp}{\frmvar}\label{pro:cca:line4}
      \ENDPAR
	\ENDIFii
	\STATE$+$ \cond{$\idle$}\label{pro:cca:line6}%
	\IFempty
      \UPD{$\timeout:=\now+\difs$}\COMMENT{start wait for duration $\difs$}				\label{pro:cca:line7}
      \PAR
	\STATE\cond{$\neg\idle$}\ \procCCA{\myip}{\backoff}{\backoffexp}{\frmvar}	\label{pro:cca:line9}
	\STATE$+$ \cond{$\idle \wedge \now \geq \timeout$}				\label{pro:cca:line10}
		\IFempty
		\UPD{\timeout:=\now+\backoff}							\label{pro:cca:line11}
		\PAR
		  \STATE\cond{$\neg\idle$}\COMMENT{busy during backoff time}\label{pro:cca:line13}%
		  \IFempty
		    \UPD{\backoff:=\timeout-\now}\ \procCCA{\myip}{\backoff}{\backoffexp}{\frmvar}	\label{pro:cca:line14}
			\ENDIFii
	\STATE$+$ \cond{$\idle \wedge \now \geq \timeout$}\COMMENT{idle for backoff time}\label{pro:cca:line15}%
	\IFempty
		    \STATE\transmit{\frmvar}\ .								\label{pro:cca:line16}
		    \procAckRecvL{\myip}{\backoffexp}{\now{+}\maxAckWait}{\frmvar}			\label{pro:cca:line17}
		  \ENDIFii
		\ENDPAR
		\ENDIFii
      \ENDPAR
  \ENDIFii

		\end{algorithmic}
    }%end{footnotesize}
  \end{algorithm}
        \vspace{-2.65ex}

First, the process senses the medium and awaits the point in time when it is idle (\Line{cca}{6}). 
In case, before this happens, it receives from another node in the network a CSMA
message destined for itself (\Line{cca}{1}), this message is handled just as in  \Proc{p_csma},
except that after acknowledging this message the protocol returns to \Proc{p_cca}.

To guarantee a gap between messages sent via the medium, CSMA/CA (as well as other protocols)
specifies the \emph{distributed (coordination function) interframe space} ($\difs\in\tTIME$), which is usually small,%
\footnote{Recommended values for the constant $\difs$ are given in \cite{IEEE80211}.}
but larger than $\sifs$, so that acknowledgements get priority over new data frames.
When the medium becomes busy during the interframe space, another node started transmitting and 
the process goes back to listening to the medium (\Line{cca}{9}). 
In case nothing happens on the medium and the end of the interframe space is reached (\Line{cca}{10}),
the process determines the actual time to start transmitting the message, taking the backoff time $\backoff$  into account (\Line{cca}{11}).  
If the medium is idle for the entire backoff period (\Line{cca}{15}), the message is transmitted  (\Line{cca}{16}),
and the process calls the process $\ACKRECV$ that will await an acknowledgement from the recipient of $\frmvar$ (\Line{cca}{17}); 
the third argument specifies the maximum time the process should wait for such an acknowledgement. 
(As mentioned before an acknowledgement may never arrive.) 
If another 
node transmits on the medium during the backoff period, the protocol restarts the routine (Lines \ref{pro:cca:line13} and \ref{pro:cca:line14}), with an adjusted backoff value $\backoff$---the process already started waiting and should not be punished when the waiting is restarted; this update guarantees fairness of the protocol.

%%%%%%%%%%%%%%%%
%% Process 4
The process awaiting an acknowledgement (\Proc{p_ackrecv}) is straightforward.
It waits until either it receives
a CSMA message destined for itself (\Line{ackrev}{1}), or it receives an acknowledgement (\Line{ackrev}{6}),
or it has waited for this acknowledgement as long as it is going to (\Line{ackrev}{10}).

In the first case, the message is handled just as in  \Proc{p_csma},
except that after acknowledging this message the protocol returns to \Proc{p_ackrecv}.
In the second case the \application\ layer is informed that the
sending of $\frmvar$ was successful and the process loops back to \Proc{p_csma}
(\Line{ackrev}{8}). \Line{ackrev}{10} describes the situation where no acknowledgement message
arrives and the process times out. Here CSMA/CA retries to send the message; the counter
$\backoffexp$ is incremented.

  \vspace*{-2.65ex}
  \algsetup{linenodelimiter=.,linenosize=\tiny}
  \begin{algorithm}[H]
    {\footnotesize
      \caption{\small Receiving an ACK}
      \label{pro:p_ackrecv}
      \begin{algorithmic}[1]
        	% !TEX root = ../all.tex
%%%%%%%%%%%%%%%%%%
\DEFPROCESS{\ACKRECV}{\myip\comma\backoffexp\comma\acktimeout\comma\frmvar}
  \STATE\cond{$\new(\dataframe{\data}{\source}{\myip})$}\ \deliver{\data}\ .	\label{pro:ackrev:line1}
  \IFempty
      \PAR                                                              \label{pro:ackrev:line2}
	\UPD{\timeout:=\now+\sifs}			                \label{pro:ackrev:line3}
	\cond{$\now \geq \timeout$}			                \label{pro:ackrev:line3b}
 	\STATE\transmit{\ackframe{\source}}\ .\ \procAckRecv{\myip}{\backoffexp}{\acktimeout}{\frmvar}    \label{pro:ackrev:line4}
      \ENDPAR
      \ENDIFii
    \STATE$+$ \cond{$\new(\ackframe{\myip})$}\COMMENT{acknowledgement received}\label{pro:ackrev:line6}%
    \IFempty
			\STATE\deliver{\returnSuccess}\ .\ \procCSMA{\myip}			\label{pro:ackrev:line8}
  \ENDIFii   
  \STATE$+$ \cond{$\now \geq \acktimeout$}\ \procInit{\myip}{\backoffexp{+}1}{\frmvar}		\label{pro:ackrev:line10}

		\end{algorithmic}
    }%end{footnotesize}
  \end{algorithm}
        \vspace{-2.65ex}

\subsection{The Hidden Station Problem}\label{sec:Hidden}
As mentioned in the introduction to this section, CSMA/CA suffers from the hidden station problem.
This refers to the situation where two nodes $A$ and $C$ are not within transmission range of each
other, while a node $B$ is in range of both. 
In this situation $C$ may be transmitting to $B$, but $A$ is not able to sense this, and thus may
start a transmission to $B$ at roughly the same time, leading to data collisions at $B$.

While CSMA/CA is not able to avoid such collisions as a whole---it is always possible that two (or more)
nodes hidden from each other happen to (randomly) choose the same backoff time to send messages%
---it is the exponential growth of the backoff slots that makes the problem
less pressing in the long run, as the following theorem shows.

\begin{theorem}\rm
  If $\maxRetransmit{=}\infty$
  then weak packet delivery holds with probability 1.
\vspace{-1ex}
\end{theorem} 
\begin{proofsketch}
Since the number of messages that nodes
transmit is bounded, and all nodes select random
times to start transmitting out of an increasing longer time span, with probability 1 each message
will eventually go through.
\hfill$\Box$
\end{proofsketch}
In practice, $\maxRetransmit$ is set to a value that is not high enough to approximate the idea
behind the above proof. In fact, the transmission time of a single message may be larger than the
maximal backoff period allowed. For this reason  the hidden station problem
does occur when running the CSMA/CA protocol, as studies have shown~\cite{Comer09}.
Nevertheless, the above analysis still shows that link layer protocols can be formally analysed by
process algebra in general, and \all\ in particular.

\subsection{A Formal Model for CSMA/CA with Virtual Carrier Sensing}\label{sec:csma_rts}
To overcome the hidden station problem the usage of a request-to-send/clear-to-send (\rts/\cts) handshaking \cite{IEEE80211} mechanism is available. 
This mechanism is also known as \emph{virtual carrier sensing}.
The exchange of \rts/\cts messages happens just before the actual data is sent, see \Fig{rtscts}.
\begin{figure}[t]\centering
\includegraphics{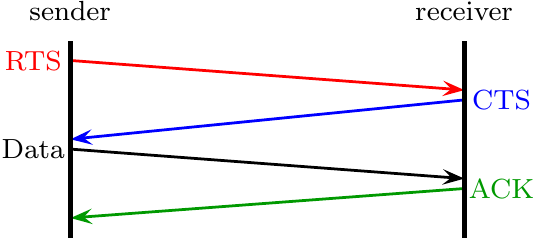}
\vspace{-4mm}
\caption{\rts/\cts exchange}\label{fig:rtscts}
\vspace{-4mm}
\end{figure}
The mechanism serves two purposes: 
(a) As the \rts and \cts messages are very short---they only contain two node identifiers as well as a natural number indicating the time it will take to send the actual \data\ (plus overhead)---the likelihood of a collision is reduced.
(b) While the handshaking does not help with solving the hidden station problem for the \rts message itself, it avoids the problem for the sending of \data. The reason is that a hidden node, which could interfere with the sending of \data\ will receive the \cts\ message from the designated recipient of \data\, and the hidden node will remain silent until the \data\ has been sent.

As for the CSMA/CA protocol we have modelled this extension in \all, based on the model of CSMA/CA we presented earlier.

Our extended model uses two functions to generate $\rtsframeID$ and $\ctsframeID$ messages, respectively.
The signature of both is $\tIP\times\tIP\times\tTIME\to\tMSG$. The first argument carries the sender (source) of the message, the second the indented destination, and the third argument 
a duration (time period) of silence that is requested/granted.
{%
\newcommand{\myplus}{\hspace{.95pt}{+}\hspace{.95pt}}%
For example, before the message $\rtsframe{\source}{\dest}{\keyw{d}}$ is transmitted, the time period \keyw{d} is calculated by\\
\centerline{\small
$\update{\keyw{d} := \sifs \myplus \durCTS \myplus \sifs \myplus \dur{\dataframe{\data}{\myip}{\dest}}\myplus\sifs\myplus\durAck}$\;\!.
}\pagebreak[3]\vspace{1mm}

\noindent
The calculation is straightforward as it follows the protocol logic and determines the amount of time needed until the acknowledgement would be 
\begin{figure}[t]\centering
\includegraphics[scale=0.85]{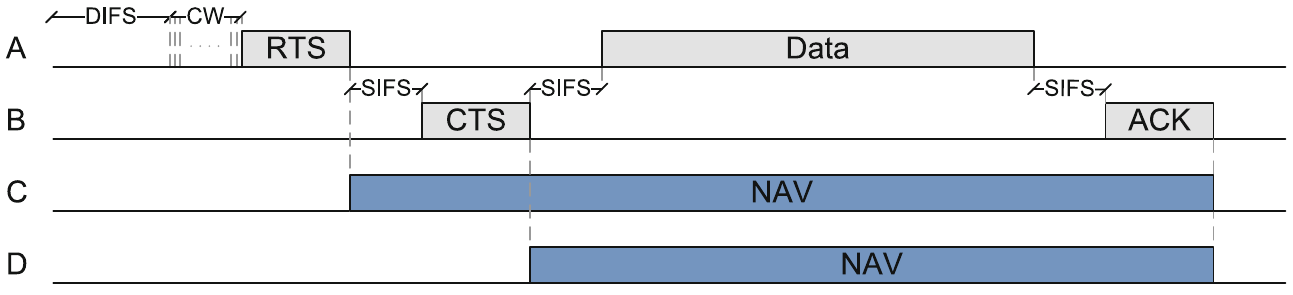}
\vspace{-4mm}
\caption{The use of virtual channel sensing using CSMA/CA~\cite{Calder14}}\label{fig:caderrts}
\vspace{-3mm}
\end{figure}
received (see \Fig{caderrts}). 
After the $\rtsframeID$ message has been received the medium should be idle for the interframe space $\sifs$;
then a $\ctsframeID$ message is sent back, which takes time $\durCTS$;
then another interframe space is needed, followed by the actual transmission of the message---the sending will take $\dur{\dataframe{\data}{\myip}{\dest}}$ time units; 
after the message is received (hopefully) another  interframe space is required before the acknowledgement is sent back.
}

Process \ref{pro:p_init} remains essentially unchanged;
it is merely equipped with the destination $\dest$ of the message that needs to be
transmitted, and an additional timed variable $\nav\in\tTIME$.
These variables are not used in this process, but required later on.
Variable $\nav$ holds the point in time until the process should not transmit any $\rtsframeID$ or $\ctsframeID$ message.
This period of silence is necessary as the node figures out that until time $\nav$ another node will transmit message(s).%
\footnote{After a successful RTS/CTS exchange, communicating nodes proceed with transmitting the
  data and an acknowledgement regardless of the value of $\nav$.}

\Proc{rts_csma} is the modified version of \Proc{p_csma}.
Identical to \Proc{p_csma} it awaits an instruction from the \application layer, or an incoming CSMA
message destined for itself.
Lines~\ref{pro:rts_csma:line1}--\ref{pro:rts_csma:line3} are identical to \Proc{p_csma}.
Lines~\ref{pro:rts_csma:line4}--\ref{pro:rts_csma:line11} handle the two new message types.
In case an $\rtsframeID$ message $\rtsframe{\source}{\dest}{\keyw{d}}$ is received that is intended
for another recipient ($\dest\neq\myip$)
the node concludes that another node wants to use the medium for the amount of $\keyw{d}$ time units; 
the process updates the variable $\nav$ if needed, indicating the period the node should remain silent, by taking the maximum of the current value of $\nav$, 
and $\now{+}\keyw{d}$, the point in time until the sender $\source$ of the $\rtsframeID$ message
requires the medium. The same behaviour occurs if a 
$\ctsframeID$ message is received that is not intended for the node itself (\Line{rts_csma}{4}).
If the incoming message is an $\rtsframeID$ message intended for the node itself (\Line{rts_csma}{6})
by default the node answers with a clear-to-send message back to the sender (\Line{rts_csma}{9}).
However, when the receiver of the $\rtsframeID$ has knowledge about other nodes requiring the medium ($\now\leq\nav$),
a clear-to-send cannot be granted, and the request is dropped (\Line{rts_csma}{6}).
\begin{figure}[t]
  \vspace{-2ex}
  \vspace*{-2.65ex}
  \algsetup{linenodelimiter=.,linenosize=\tiny}
  \begin{algorithm}[H]
    {\footnotesize
      \caption{\small The Basic Routine (\rts/\cts)}
      \label{pro:rts_csma}
      \begin{algorithmic}[1]
        	\input{processes/rts_csma.tex}
		\end{algorithmic}
    }%end{footnotesize}
  \end{algorithm}
        \vspace{-2.65ex}

  \vspace{-9.7ex}
\end{figure}
Similar to the sending of an acknowledgement (\Line{rts_csma}{2}), the process waits for the short
interframe space (\sifs) before sending the CTS  (\Line{rts_csma}{6}). 
\begin{figure}[t]
  \vspace{-2ex}
  \vspace*{-2.65ex}
  \algsetup{linenodelimiter=.,linenosize=\tiny}
  \begin{algorithm}[H]
    {\footnotesize
      \caption{\small Clear Channel Assessment With Virtual Carrier Sense}
      \label{pro:rts_cca}
      \begin{algorithmic}[1]
        	\input{processes/rts_cca.tex}
		\end{algorithmic}
    }%end{footnotesize}
  \end{algorithm}
        \vspace{-2.65ex}

  \vspace{-5ex}
\end{figure}
\Line{rts_csma}{8} handles the case where the medium becomes busy ($\neg\idle$) during this period; also
here a clear-to-send cannot be granted, and the request is dropped.%
\footnote{The condition $\now>\timeout{-}\sifs$ prevents the process from dropping the request in
  the very first time slice that \CSMA\ is running. Here the medium counts as busy, but only because
  we have just received an $\rtsframeID$ message.}
Only when the medium stays idle during the entire interframe space
the node $\myip$ can inform the source of the $\rtsframeID$ message that the medium is clear to
send; the $\ctsframeID$ is transmitted in \Line{rts_csma}{9}. The time a receiver of this
message has to be silent is adjusted by deducting the time elapsed before this happens.\pagebreak[2]
In \Line{rts_csma}{10} the process resets $\nav$ to remind itself not to issue any $\rtsframeID$
message until the present exchange has been completed.%
\footnote{A case $\new(\ctsframe{\source}{\dest}{\keyw{d}})\wedge \dest=\myip$ is not required as a $\ctsframeID$ message is only expected in case an $\rtsframeID$ was sent, and hence handled in process \CTSSEND.}

\Proc{rts_cca} is the modified version of \Proc{p_cca}.
The goal of this process is to send an $\rtsframeID$ message (\Line{rts_cca}{22}).
Before it can start its work, it waits until the medium is idle, and
any time it is required to be silent has elapsed (\Line{rts_cca}{11}).
Until this happens incoming data frames, $\rtsframeID$ or
$\ctsframeID$ messages are treated just as in \Proc{rts_csma}: Lines~\ref{pro:rts_cca:line1}--\ref{pro:rts_cca:line10}
copy Lines~\ref{pro:rts_csma:line2}--\ref{pro:rts_csma:line11} of \Proc{rts_csma}, except that
afterwards the process returns to itself.
Then Lines~\ref{pro:rts_cca:line12}--\ref{pro:rts_cca:line20} are copied from
Lines~\ref{pro:cca:line7}--\ref{pro:cca:line15} from \Proc{p_cca}.
\Line{rts_cca}{21} calculates the time other nodes ought to keep silent when receiving the $\rtsframeID$
message, and \Line{rts_cca}{23} passes control to the process \CTSRECV, which awaits a $\ctsframeID$
response to the $\rtsframeID$ message transmitted in \Line{rts_cca}{22}.
The fourth argument of \CTSRECV\ specifies the maximum time that process should wait for such a response;
a good value for $\maxCtsWait$ is $\sifs + \durCTS$.

Process \CTSRECV\ listens for this time to a $\ctsframeID$ message with source $\dest$ and
destination $\myip$. In case the expected $\ctsframeID$ message arrives in time (\Line{rts_ctsrec}{2}),
the node waits for a time $\sifs$ (\Line{rts_ctsrec}{3})\pagebreak[4] and then transmits the data frame and
proceeds to await an acknowledgement (\Line{rts_ctsrec}{4}).
The fourth argument of \ACKRECV\ specifies the maximum time the process should wait for such an acknowledgement;
a good value for $\maxAckWait$ is $\sifs + \durAck$. 
If the $\ctsframeID$ message does not arrive in time (\Line{rts_ctsrec}{5}),
the process returns to \INIT\ to send another $\rtsframeID$ message, while incrementing the counter $\backoffexp$ (\Line{rts_ctsrec}{6}).
While waiting for the $\ctsframeID$ message, any incoming $\rtsframeID$ or $\ctsframeID$ message
destined for another node is treated exactly as in \Proc{rts_csma} (Lines~\ref{pro:rts_ctsrec:line4b}--\ref{pro:rts_ctsrec:line4c}).
Incoming data frames cannot arrive when this process is running, and incoming $\rtsframeID$ messages
to $\myip$ are ignored.

  \vspace*{-2.65ex}
  \algsetup{linenodelimiter=.,linenosize=\tiny}
  \begin{algorithm}[H]
    {\footnotesize
      \caption{\small Receiving a CTS}
      \label{pro:rts_ctsrec}
      \begin{algorithmic}[1]
        	% !TEX root = ../all.tex
%%%%%%%%%%%%%%%%%%
\DEFPROCESS{\CTSRECV}{\myip\comma\dest\comma\backoffexp\comma\keyw{ctstimeout}\comma\frmvar\comma\nav}
	\STATE\cond{$\new(\ctsframe{\dest}{\myip}{\keyw{d}})$}                    \label{pro:rts_ctsrec:line2}
	\IFempty
		\UPD{\timeout:= \now+\sifs}                                             \label{pro:rts_ctsrec:line3}
		\cond{$\now\geq\timeout$}                                               \label{prcln:wait-for-dataframe}
		\STATE\transmit{\frmvar}\ .                                             \label{pro:rts_ctsrec:line4}
		\procAckRecv{\myip\comma\dest}{\backoffexp}{\now+\maxAckWait\comma\frmvar}{\nav}
		\ENDIFii
		\STATE$+$ \cond{$(\new(\rtsframe{\source}{\dest}{\keyw{d}}) \vee \new(\ctsframe{\source}{\dest}{\keyw{d}})) \wedge \dest \neq \myip \wedge \nav < \now\mathord+\keyw{d}$}\label{pro:rts_ctsrec:line4b}%
		\IFempty
 	\STATE\textcolor{brown}{\assignment{\nav:= \now{+}\keyw{d}}}\
	\procCtsRecv{\myip\comma\dest}{\backoffexp}{\keyw{ctstimeout}\comma\frmvar}{\nav} \label{pro:rts_ctsrec:line4c}
	\ENDIFii
	\STATE$+$ \cond{$\now \geq \keyw{ctstimeout}$}                                           \label{pro:rts_ctsrec:line5}%
	\IFempty
		\STATE \procInit{\myip\comma\dest}{\backoffexp{+}1}{\frmvar\comma\nav}  \label{pro:rts_ctsrec:line6}
	\ENDIFii

		\end{algorithmic}
    }%end{footnotesize}
  \end{algorithm}
        \vspace{-2.65ex}

\mbox{}\vspace{-15pt}

  \vspace*{-2.65ex}
  \algsetup{linenodelimiter=.,linenosize=\tiny}
  \begin{algorithm}[H]
    {\footnotesize
      \caption{\small Receiving an ACK}
      \label{pro:rts_ackrecv}
      \begin{algorithmic}[1]
        	% !TEX root = ../all.tex
%%%%%%%%%%%%%%%%%%%
\DEFPROCESS{\ACKRECV}{\myip\comma\dest\comma\backoffexp\comma\acktimeout\comma\frmvar\comma\nav}
	\STATE\cond{$\new(\ackframe{\myip})$}                    \label{pro:rts_ackrecv:line2}
	\IFempty
		\deliverL{\returnSuccess}\ .
		\procCSMA{\myip\comma\nav}                             \label{pro:rts_ackrecv:line3}
		\ENDIFii
		\STATE$+$ \cond{$(\new(\rtsframe{\source}{\dest}{\keyw{d}}) \vee \new(\ctsframe{\source}{\dest}{\keyw{d}})) \wedge \dest \neq \myip
    \wedge \nav < \now+\keyw{d}$}\label{pro:rts_ackrecv:line4a}%
    \IFempty
 	\STATE\textcolor{brown}{\assignment{\nav:= \now{+}\keyw{d}}}\
        \procAckRecv{\myip\comma\dest}{\backoffexp}{\acktimeout\comma\frmvar}{\nav}	\label{pro:rts_ackrecv:line4b}
	\ENDIFii
		\STATE$+$ \cond{$\now \geq \acktimeout$}\COMMENT{nothing received}\label{pro:rts_ackrecv:line4}%
		\IFempty
		\STATE \procInit{\myip\comma\dest}{\backoffexp{+}1}{\frmvar\comma\nav}  \label{pro:rts_ackrecb:line5}
	\ENDIFii

		\end{algorithmic}
    }%end{footnotesize}
  \end{algorithm}
        \vspace{-2.65ex}

\Proc{rts_ackrecv} handles the receipt of an acknowledgement in response to a successful data transmission.
If an acknowledgement arrives, it must be
from the node to which $\myip$ has transmitted a data frame. In that case (\Line{rts_ackrecv}{2}),
the \application\ layer is informed that the sending of $\frmvar$ was successful and the process loops back to \Proc{rts_csma}
(\Line{rts_ackrecv}{3}). \Line{rts_ackrecv}{4} describes the situation where no acknowledgement message
arrives and the process times out. Also here CSMA/CA retries to send the message; the counter
$\backoffexp$ is incremented. 
Lines~\ref{pro:rts_ackrecv:line4a}--\ref{pro:rts_ackrecv:line4b} describe the usual handling of 
incoming $\rtsframeID$ or $\ctsframeID$ messages destined for another node.

\subsection{The Exposed Station Problem}

Another source of collisions in CSMA/CA is the well-known \emph{exposed station problem}.
This refers to a linear topology $A - B - C - D$, where an unending stream of messages between
$C$ and $D$ interferes with attempts by $A$ to get a message across to $B$.
In the default CSMA/CA protocol as formalised in \Sect{csma-default}, transmissions from $A$ to $B$ may
perpetually collide at $B$ with transmissions from $C$ destined  for $D$.
CSMA/CA with virtual carrier sensing mitigates this problem, for a {\ctsframeID} sent by $B$ in response to
an {\rtsframeID} sent by $A$ will tell $C$ to keep silent for the required duration.
In fact, we can show that in the above topology, if $\maxRetransmit{=}\infty$
then packet delivery holds with probability 1.
A non-probabilistic guarantee cannot be given since nodes $A$ and $C$ could behave
in the same way, meaning
if one node is sending out a message the other does the same at the very same moment, and if one is silent the other remains silent as well. 
In this scenario all messages to be sent are doomed.

Based on our formalisation, we can prove that once the \rts/\cts handshake has been successfully
concluded, meaning that all nodes within range of the intended recipient have received the {\ctsframeID},
then packet delivery holds outright. So the only problem left is to achieve a successful \rts/\cts handshake.
Since {\rtsframeID} and {\ctsframeID} messages are rather short, even by modest values of $\maxRetransmit$ it becomes
likely that such messages do not collide.

In spite of this, CSMA/CA with (or without) virtual channel sensing cannot achieve packet delivery
with probability 1 for general topologies.
	Assume the following network topology
	\vspace{-2mm}
\begin{center}	
	\begin{tikzpicture}
		\draw (1,1) -- (2,1.8) -- (3,1.8);
		\draw (0,1) -- (1,1) -- (2,1) -- (3,1);
		\draw (1,1) -- (2,0.2) -- (3,0.2);
		\draw[fill=white] (0,1) circle (8pt);\draw (0,1) node{$B$};
		\draw[fill=white] (1,1) circle (8pt);\draw (1,1) node{$A$};
		\draw[fill=white] (2,1.8) circle (8pt);\draw (2,1.8) node{$C_1$};
		\draw[fill=white] (3,1.8) circle (8pt);\draw (3,1.8) node{$D_1$};
		\draw[fill=white] (2,1) circle (8pt);\draw (2,1) node{$C_2$};
		\draw[fill=white] (3,1) circle (8pt);\draw (3,1) node{$D_2$};
		\draw[fill=white] (2,0.2) circle (8pt);\draw (2,0.2) node{$C_3$};
		\draw[fill=white] (3,0.2) circle (8pt);\draw (3,0.2) node{$D_3$};
	\end{tikzpicture}
\end{center}

Here it may happen that one of the $C_i$s is always busy transmitting a large message to $D_i$;\pagebreak[3]
any given $C_i$ is occasionally silent (not sending any message), but then one of the others is transmitting.
As $C_i$ is disconnected from $C_j$, for $j\neq i$, coordination between the nodes is impossible.
As a consequence, the medium at $A$ will always be busy, 
so that $A$ cannot send an {\rtsframeID} message to $B$.

%%%%%%%%%%%%%%%%%%%%%%%%%%%%%%%%%%
\section{Related Work}
The CSMA protocol in its different variants has been analysed with different formalisms in the past. 

Multiple analyses were performed for the CSMA/CD protocol (CSMA with collision detection), a predecessor of CSMA/CA that has a constant backoff, i.e.\ the backoff time is not increased exponentially, see 
\cite{Duflot04,DKN+13,Zhao04,JensenLS96,Parrow88}.
In all these approaches frame collisions have to be modelled explicitly, as part of the protocol description.
In contrast, our approach handles collisions in the semantics; thereby achieving a clear separation between protocol specifications and  link layer behaviour.

Duflot et al.\ \cite{Duflot04,DKN+13} use probabilistic timed automata (PTAs) to model the protocol, and use 
probabilistic model checking (\prism) and approximate model checking (\apmc) for their analysis.
The model explained in~\cite{Zhao04} is based on PTAs as well, 
but uses the model checker \uppaal as verification tool. 
These approaches, although formal, have very little in common with our approach. 
On the one hand it is not easy to change the model from CSMA/CD to CSMA/CA, as the latter requires unbounded data structures (or alike) to model the exponential backoff. On the other hand, as usual, model checking suffers from state space explosion and only small networks (usually fewer than ten nodes) can be analysed. This is sufficient and convenient when it comes to finding counter examples, but these approaches cannot provide guarantees for arbitrary network topologies, as ours does.

Jensen et al.\ \cite{JensenLS96} use models of CSMA/CD to compare the tools SPIN and \uppaal. 
Their models are much more abstract than ours.
It is proven that no collisions will ever occur, without stating the exact conditions under which this statement holds.

To the best of our knowledge, Parrow \cite{Parrow88} is the only one who uses  process algebra (CCS) to model and analyse CSMA. 
His untimed model of CSMA/CD is extremely abstract and the analysis performed is limited to two nodes only, avoiding scenarios such as the hidden station problem.

There are far fewer formal analyses techniques available when it comes to CSMA/CA (with and without virtual medium sensing).
Traditional approaches to the analysis of network protocols are simulation and test-bed experiments.
This is also the case for CSMA/CA (e.g.\ \cite{ChhayaGupta97}). While these are important and valid methods for protocol evaluation, in particular for quantitative performance evaluation, they have limitations in regards to the evaluation of basic protocol correctness properties.

\hspace{-3pt}Following the spirit of the above-mentioned research of model checking CSMA, Fruth~\cite{Fruth06} analyses CSMA/CA using PTAs and \prism. 
He considers properties such as the 
minimum probability of two nodes successfully completing their transmissions, and maximum expected number of collisions until two 
nodes have successfully completed their transmissions. 
As before, this analysis technique does not scale; in \cite{Fruth06} the experiments are limited to two contending nodes only.

Beyond model checking, simulation and test-bed experiments, we are only aware of two other formal approaches.
In \cite{Bianchi00} Markov chains are used to derive an accurate, analytical model to compute the throughput of CSMA/CA.
Calculating throughput is an orthogonal task to our vision of proving (functional) correctness.

An approach aiming at proving the correctness of CSMA/CA with virtual carrier sensing (\rts/\cts), and hence related to ours, is presented in 
\cite{Calder14}. Based on stochastic bigraphs with sharing it uses rewrite rules to analyse quantitative properties. 
Although it is an approach that is capable to analyse arbitrary topologies, to apply the rewrite rules a particular topology 
needs to be modelled by a directed acyclic graph structure, which is part of the bigraph.

%%%%%%%%%%%%%%%%%%%%%%%%%%%%%%%%%%

\section{Conclusion}
In this paper we have proposed a novel process algebra, called \all, that can be used to model, verify and analyse link layer protocols.
Since we aimed at a process algebra featuring aspects of the link layer such as frame collisions, as well as arbitrary data structures (to model a rich class of protocols), we could not use any of the existing algebras. 
The design of \all is layered. The first layer allows modelling protocols in some sort of pseudo code, which hopefully makes our approach accessible for network and software researchers/engineers. The other layers are mainly for giving a formal semantics to the language. The layer of partial network expressions, the third layer, provides a unique and sophisticated mechanism for modelling the collision of {\dpkt}s. As it is hard-wired in the semantics there is no need to model collisions manually when modelling a protocol, as it was done before \cite{Parrow88}. 
Next to primitives needed for modelling link layer protocols (e.g.\ \textbf{transmit})
and standard operators of process algebra (e.g.\ nondeterministic choice), \all\ provides an operator for probabilistic choice. 

This operator is needed to model aspects of link layer protocols such as the exponential backoff for the Carrier-Sense Multiple Access with Collision Avoidance protocol, the case study we have chosen to demonstrate the applicability of \all.
We have modelled and analysed two versions of CSMA/CA, without and with virtual carrier sensing.
Our analysis has confirmed the hidden station problem for the version without virtual carrier sensing. 
However, we have also shown that the version with virtual carrier sensing overcomes not only this problem, but also the exposed station problem with probability 1. 
Yet the protocol cannot guarantee packet delivery, not even with probability 1.

To perform this analysis we had to formalise suitable liveness properties for link layer protocols specified in our framework. 

\subsubsection*{Acknowledgement:} We thank Tran Ngoc Ma for her involvement in this project in a very early phase. 
We also like to thank the German Academic Exchange Service (DAAD) that funded an internship of the third author at 
Data61, CSIRO.
%%%%%%%%%%%%%%%%%%%%%%%%%%%%%%%%%%
\bibliographystyle{splncs03}
\bibliography{csma}
%%%%%%%%%%%%%%%%%%%%%%%%%%%%%%%%%%

\end{document}